\documentclass{myreport}

\usepackage{myreport}

\title{
    Lower bounds on the spatial decay of 
    the ground state of the Pauli-Fierz model
}
\author{
    Fumio Hiroshima\footnote{Faculty of Mathematics, Kyushu University, Fukuoka 819--0395, Japan} \,
    and Yuki Tsujimoto\footnote{Joint Graduate School of Mathematics for Inovation, Kyushu University, Fukuoka 819--0395, Japan}
}
\date{
    \today
}

\begin{document}
\maketitle

% 概要
\begin{abstract}
    Lower bounds on the pointwise spatial decay of the ground state of the Pauli-Fierz model in non-relativistic quantum electrodynamics are studied.
    Two models are considered: the full Pauli-Fierz model and the Pauli-Fierz model under the dipole approximation.
    For the full model, a lower bound is derived by means of a probabilistic approach based on the Feynman-Kac formula.
    The application of the Agmon distance method to the full model encounters a serious difficulty arising from the path dependence of a double stochastic integral.
    In contrast, under the dipole approximation, the corresponding integrand becomes deterministic, which allows a lower bound to be obtained in terms of the Agmon distance.
\end{abstract}

%\tableofcontents

\section{Introduction}

In this paper, we derive lower bounds on the point-wise spatial decay of the ground state of the Pauli-Fierz model in non-relativistic quantum electrodynamics.
The Pauli-Fierz model was originally introduced in~\cite{PF38} and describes the interaction between an electron and a quantized radiation field.
In this model, electrons are treated as non-relativistic particles.
On the other hand, photons are described by a quantized radiation field  
in the Coulomb gauge, in which only transverse modes are present.
The decoupled Hamiltonian is given by 
\[\left(-\frac12 \Delta+V \right)\otimes \mathbbm{1}+\mathbbm{1}\otimes H_\mathrm{rad},\]
where $-\tfrac12 \Delta+V$ describes the Hamiltonian of the non-relativistic particles and 
$H_\mathrm{rad}$ the free field hamiltonian. 
The interaction between the quantized radiation field $A=(A_1,A_2,A_3)$ and electrons is given by the  minimal coupling: 
$-i\nabla_\mu\otimes \mathbbm{1} \to -i\nabla_\mu\otimes\mathbbm{1}-\alpha A_\mu$, where $\alpha\in\mathbb{R}$ is the coupling constant.  
Then the Hamiltonian is given by
\begin{align}\label{eq:PF} %
    H_\alpha
    = \frac{1}{2} {(- i \nabla \otimes \mathbbm{1} - \alpha A)}^2
    + V \otimes \mathbbm{1}
    + \mathbbm{1} \otimes H_\mathrm{rad} .
\end{align}
The operator \cref{eq:PF} acts on the Hilbert space: 
\begin{align*}
    L^2(\mathbb{R}^3) \otimes \mathcal{F}
    \simeq \int_{\mathbb{R}^3}^\oplus \mathcal{F} \, \mathrm{d}x ,
\end{align*}
where $\mathcal{F}$ denotes the boson Fock space associated with the quantized radiation field and 
the right-hand side above denotes the constant fiber direct integral~\cite[\text{X\hspace{-1.2pt}I\hspace{-1.2pt}I\hspace{-1.2pt}I.6}]{rs4}.
In order to define the quantized radiation field
\begin{align*}
    A = \int_{\mathbb{R}^3}^\oplus A(x) \, \mathrm{d}x ,
\end{align*}
we impose an ultraviolet cutoff function $\hat{\varphi}$. 
Spectral analysis of the Pauli-Fierz type Hamiltonian has accelerated since the late 1990s.
A comprehensive treatment can be found in e.g., \cite{spo04}.
Since the bottom of the spectrum of the Hamiltonian is embedded in the continuous spectrum since photons are massless, proving the existence and uniqueness of a ground state is a nontrivial problem.
Finally the existence of the ground state of the Pauli-Fierz Hamiltonian has been established in~\cite{bfs3, GLL01}.
See~\cite{AH97, BFS1, BFS2, spo98, ger00,LMS07, GHPS11} for related results, and  
\cite{BSS25,CEH04,CVV03,HS01b,HVV03,KM13b} for the enhanced binding. 
The spectrum of the semi-relativistic Pauli-Fierz 
\[\sqrt{(-i\nabla-\alpha A)^2+m^2}-m+V\otimes\mathbbm{1}+\mathbbm{1}\otimes H_\mathrm{rad}\]
is also studied in \cite{MS09a,KMS11a,KMS11b,KM13a,hir14,HH16,HHS21}. 

Having established the existence of a ground state, we now turn to the investigation of its spatial behavior.
In~\cite{HH10}, upper bounds on the spatial decay of the ground state were derived by means of the Feynman-Kac formula.
Let 
$\varphi_\mathrm{g}$ be the normalized ground state of the Pauli-Fierz Hamiltonian. 
Our aim here is to establish a lower bound on $\norm{{\varphi_\mathrm{g}}(x)}_\mathcal{F}$.

Although our analysis also relies on the Feynman-Kac formula, the approach taken in the present paper differs substantially from that of~\cite{HM21}, where 
upper and lower bounds of the spacial decay of the ground state of the renormalized Nelson Hamiltonian \cite{nel64a} is given.   
Our argument rests on two key ingredients:
\begin{enumerate}
    \item[(1)] the positivity property of ${(\mathbbm{1}, {\varphi_\mathrm{g}}(x))}_\mathcal{F}$;
    \item[(2)] an estimate on the exponential moment $\mathbb{E}[e^{\varepsilon X}]$ of the double stochastic integral $X$.
\end{enumerate}
The implementation of this strategy, however, encounters two principal difficulties.
The first arises from the {phase factor} appearing in the Feynman-Kac formula, while the second concerns the control of the exponential moment of a formal {double stochastic integral}.
We now explain how these two difficulties are overcome.

{(Phase rotation)}
We apply the identity
\begin{align*}
    \varphi_\mathrm{g}
    = e^{- T (H_\alpha - E_\alpha)} \varphi_\mathrm{g}
\end{align*}
and the Feynman-Kac formula:
\begin{align*}
   e^{- T (H_\alpha - E_\alpha)} \Psi(x)
    = 
        \mathbb{E}^{x}\left[
            e^{- \int_0^T (V(B_t)-E_\alpha) \, \mathrm{d}t}
                \mathrm{J}_0^* e^{- i \alpha {A_\mathrm{E}}(K_T)} \mathrm{J}_T
                \Psi(B_T) %
        \right]. 
\end{align*}
Here ${(B_t)}_{t \geq 0}$ denotes $3$-dimensional Brownian motion. 
The phase factor
\begin{align*}
    e^{- i \alpha {A_\mathrm{E}}(K_T)}
\end{align*}
appearing in the integrand is the main obstacle to deriving a lower bound for ${\varphi_\mathrm{g}}(x)$.
We note that $\mathbb{R}^3 \ni x \mapsto {(\mathbbm{1}, {\varphi_\mathrm{g}}(x))}_\mathcal{F}$ is continuous for Kato-decomposable potential $V$. 
The ground state $\varphi_\mathrm{g}$ of \cref{eq:PF} is {not} positive.
Although
\begin{align*}
    \norm{{\varphi_\mathrm{g}}(x)}_\mathcal{F}
    \geq \abs{{(\mathbbm{1}, {\varphi_\mathrm{g}}(x))}_\mathcal{F}} ,
\end{align*}
it is not clear whether ${(\mathbbm{1}, {\varphi_\mathrm{g}}(x))}_\mathcal{F}$ is positive, where $\mathbbm{1}$ denotes the Fock vacuum.
In~\cite{hir00b}, however, it was shown that the transformed ground state $\Theta^{-1} \varphi_\mathrm{g}$ is positive, and hence
\begin{align*}
    {(\mathbbm{1}, {\varphi_\mathrm{g}}(x))}_\mathcal{F}
    = {(
        \Theta^{-1}
        \mathbbm{1} , %
        \Theta^{-1}
        {\varphi_\mathrm{g}}(x) %
    )}_\mathcal{F}
    = {(
        \mathbbm{1} , %
        \Theta^{-1}
        {\varphi_\mathrm{g}}(x) %
    )}_\mathcal{F}
    > 0 ,
\end{align*}
where $\Theta = e^{i \frac{\pi}{2} N}$ is a phase rotation and $N$ denotes the number operator.
By the identity $e^{- T (H_\alpha - E_\alpha)} \varphi_\mathrm{g} = \varphi_\mathrm{g}$, we also have
\begin{align}\label{eq:fkf2} %
    {(\mathbbm{1}, {\varphi_\mathrm{g}}(x))}_\mathcal{F}
    = {(
        \mathbbm{1} , %
        (e^{- T (H_\alpha - E_\alpha)}
        {\varphi_\mathrm{g}})(x) %
    )}_\mathcal{F} .
\end{align}
Note that the integrand in the Feynman-Kac formula for $e^{-T (H_\alpha - E_\alpha)}$ is {complex-valued}, i.e., $e^{- i A(K)}$ appears in the integrand, and hence it is not straightforward to estimate \cref{eq:fkf2} from below.
Then the arguments used in \cite{HM21} do not readily extend to the Pauli-Fierz model.
We therefore adopt a different approach.
We have
\begin{align}\label{eq:fkf1} %
    {(
        \mathbbm{1} , %
        {\varphi_\mathrm{g}}(x) %
    )}_\mathcal{F}={(
        \mathbbm{1} , %
        \Theta^{-1}
        {\varphi_\mathrm{g}}(x) %
    )}_\mathcal{F}
    = {(
        \mathbbm{1} , %
        e^{- T (\Theta^{-1} H_\alpha \Theta - E_\alpha)} \Theta^{-1}
        {\varphi_\mathrm{g}}(x) %
    )}_\mathcal{F} .
\end{align}
The Feynman-Kac formula for $e^{- T (\Theta^{-1} H_\alpha \Theta - E_\alpha)}$ can also be constructed. More importantly, it was shown in~\cite{hir00b} that its integrand is {positive}, in contrast to the complex-valued integrand appearing in the Feynman-Kac formula for $e^{- T (H_\alpha - E_\alpha)}$.
Therefore, the right-hand side of \cref{eq:fkf1} can be estimated from below.
This positivity property is the key ingredient in our analysis.
See \cref{lem:3}.

{(Double stochastic integral and Agmon distance)}
We also have another difficulty in estimating the spatial decay of ${\varphi_\mathrm{g}}(x)$.
One standard way of the estimate of the spacial decay from below is an application of Agmon distance associated with external potential $V$:
\begin{align}\label{eq:agmon} %
    \norm{{\varphi_\mathrm{g}}(x)}_\mathcal{F}
    \geq D e^{- c \mathscr{G}(\gamma)} ,
\end{align}
where $c > 0$ and $D > 0$ are constants and
\begin{align*}
    \mathscr{G}(\gamma)
    = \int_0^T \sqrt{2 V(\gamma_t)} \abs{\dot{\gamma}_t} \, \mathrm{d}t
\end{align*}
for any path $\gamma_s$ connecting $0$ and $x$.
The Agmon distance provides a powerful framework for the analysis of spatial decay of bound states; not only for Schr\"dinger operators \cite{CS81} but also it has been successfully employed to obtain precise upper and lower bounds on the spatial decay of the ground state by the Agmon distance even for the renormalized Nelson model in~\cite{HM21}.
To obtain the form of \cref{eq:agmon} we need to estimate $\mathbb{E}[e^{\varepsilon X}]$, where $X$ is the formal double stochastic integral: 
\begin{align*}
    X
    = \sum_{\mu, \nu = 1}^3
        \int_0^T
            \int_0^T \left\{
                \int_{\mathbb{R}^3}
                    \left(
                        \delta_{\mu \nu}
                        - \frac{k_\mu k_\nu}{\abs{k}^2}
                    \right)
                    \frac{\abs{\hat{\varphi}(k)}^2}{\omega(k)}
                    e^{- i k \cdot (B_s - B_t)}
                    e^{- \abs{s - t} \omega(k)}
                \, \mathrm{d}k
            \right\} \, \mathrm{d}B_s^\mu
        \, \mathrm{d}B_t^\nu .
\end{align*}
This expression is only formal, since the integrand does not satisfy the adaptedness condition required for Ito stochastic integrals.
Consequently, estimating $\mathbb{E}[e^{\varepsilon X}]$ is far from straightforward.
In fact, obtaining such an estimate constitutes one of the main technical difficulties of the present paper. More precisely, we show in \cref{lem:6} that
\begin{align*}
    \mathbb{E}[e^{\varepsilon X}] < + \infty
\end{align*}
provided that
\begin{align*}
    \varepsilon < \frac{1}{\triangle T}
\end{align*}
for some constant $\triangle > 0$.
Thus, the admissible range of $\varepsilon$ ensuring the finiteness of $\mathbb{E}[e^{\varepsilon X}]$ depends explicitly on $T$.
It is precisely this dependence on $T$ that gives rise to the difficulty to obtain \cref{eq:agmon}.
This phenomenon is already visible in the elementary case of Brownian motion. 
Indeed, we see that $\mathbb{E}[e^{\varepsilon B_T^2}] = {(1 - 2 \varepsilon T)}^{- 3 / 2}$, and hence
\begin{align*}
    \mathbb{E}[e^{\varepsilon B_T^2}] < + \infty
    \quad \Longleftrightarrow \quad
    \varepsilon < \frac{1}{2 T} .
\end{align*}
Thus, even in this elementary example, no fixed positive $\varepsilon$ guarantees the finiteness of the exponential moment for all $T > 0$.

The Pauli-Fierz Hamiltonian under the dipole approximation is defined by replacing $A(x)$ with $A(0)$:
 \begin{align*}
    H_\alpha(0) = \frac{1}{2}{(-i\nabla \otimes \mathbbm{1}- \alpha \mathbbm{1}\otimes A(0))}^2 + V \otimes\mathbbm{1}+ \mathbbm{1}\otimes H_\mathrm{rad} .
\end{align*}
Under the dipole approximation, by contrast, the exponential moment of the double stochastic integral can be handled in a substantially simpler manner.
The essential simplification stems from the fact that the integrand becomes deterministic, thereby making it considerably easier to derive lower bounds on the spatial decay of the ground state. More precisely, $X$ is replaced by
\begin{align*}
    X_\mathrm{dip}
    = \sum_{\mu, \nu = 1}^3
        \int_0^T
            \int_0^T \left\{
                \int_{\mathbb{R}^3}
                    \left(
                        \delta_{\mu \nu}
                        - \frac{k_\mu k_\nu}{\abs{k}^2}
                    \right)
                    \frac{\abs{\hat{\varphi}(k)}^2}{\omega(k)}
                    e^{- \abs{s - t} \omega(k)}
                \, \mathrm{d}k
            \right\} \, \mathrm{d}B_s^\mu
        \, \mathrm{d}B_t^\nu .
\end{align*}

In this case, we can show that
\begin{align*}
    \mathbb{E}[e^{\varepsilon X_\mathrm{dip}}] < + \infty
\end{align*}
provided that
\begin{align*}
   \varepsilon
   < \frac{3}{8 \alpha^2 \norm{\hat{\varphi} / \omega}^2} .
\end{align*}
The crucial point is that the admissible range of $\varepsilon$ is now independent of $T$.
As a consequence, the Agmon distance method as \cref{eq:agmon} can be implemented without substantial difficulty under the dipole approximation.
See \cref{thm:20}.

For the full Pauli-Fierz model, however, the situation is fundamentally different.
The corresponding stochastic integral retains its path dependence, and the admissible range of $\varepsilon$ shrinks as $T$ increases.
This dependence on $T$ constitutes a serious obstruction to the Agmon distance method.
In particular, the estimates required for the Agmon distance argument cannot be made uniform in $T$, and the method therefore fails to yield the desired lower bound on the spatial decay of the ground state.
We discuss this difficulty in detail in \cref{sec:5}.
Thus, despite its effectiveness for the renormalized Nelson model and for the Pauli-Fierz model under the dipole approximation, the Agmon distance does not provide an effective tool for obtaining the corresponding lower bound in the full Pauli-Fierz model. 
We therefore adopt a different approach to establish the lower bound; see \cref{thm:14}.
As the corollary of Theorem \ref{thm:14}  
we show upper and lower bounds of $\norm{{\varphi_\mathrm{g}}(x)}_\mathcal{F}$. 

\begin{corollary*}[Corollary\ref{maincoro}]
Suppose that, for some real number $n>0$ and constants
$0<A\leq B$ and $c_1,c_2>0$,
\[
    A^2|x|^{2n}-c_1
    \leq V(x)
    \leq B^2|x|^{2n}+c_2,
    \qquad x\in\mathbb{R}^3.
\]
Let $0<\epsilon<b<1/2$.
Then there exist constants $D_\epsilon,C_\epsilon>0$
such that
\[
    D_\epsilon e^{-(\sqrt2 2^n B+\epsilon)|x|^{n+1}}
    \leq \|\varphi_{\mathrm g}(x)\|_{\mathcal F}
    \leq
    C_\epsilon e^{-(b-\epsilon)\frac{A}{2^{n+4}}|x|^{n+1}},
    \qquad x\in\mathbb{R}^3.
\]
\end{corollary*}

This paper is organized as follows.
In \cref{sec:2}, we introduce the Pauli-Fierz model and recall its Feynman-Kac formula.
In \cref{sec:3}, we investigate the spatial decay of the ground state of the full Pauli-Fierz Hamiltonian.
In \cref{sec:4}, we study the spatial decay of the ground state of the Pauli-Fierz Hamiltonian under the dipole approximation by means of the Agmon distance.
In \cref{sec:5}, we derive a lower bound on the ground state of the full Pauli-Fierz Hamiltonian in terms of the Agmon distance.
Finally, in Appendix, we collect and prove several fundamental lemmas needed to make the paper self-contained.

\section{The Pauli-Fierz model}\label{sec:2} %

\subsection{Boson Fock space}
In the Pauli-Fierz model, photons are described as transversal waves in the Coulomb gauge and therefore possess two polarization.
Their mass is zero.
Accordingly, the Hilbert space of a single photon is given by
\begin{align*}
    \mathfrak{h} = L^2(\mathbb{R}^3 \times \{1, 2\}) ,
\end{align*}
where $\{1, 2\}$ corresponds to the three dimensional transverse polarizations.
For each momentum $k = (k_1, k_2, k_3) \neq 0$, we fix an orthonormal basis
\begin{align*}
    {\{
        e(k, \lambda)
        = (
            {e_1}(k, \lambda) ,
            {e_2}(k, \lambda) ,
            {e_3}(k, \lambda)
        )
    \}}_{\lambda = 1, 2}
\end{align*}
of $\mathbb{R}^3$ and transverse polarization vectors satisfy
\begin{align*}
    e(k, \lambda) \cdot e(k, \lambda') &= \delta_{\lambda \lambda'} ,
    \quad
    \sum_{\lambda = 1, 2} {e_{\mu}}(k, \lambda) {e_{\nu}}(k, \lambda) = d_{\mu \nu}(k) ,
    \quad
    k \cdot e(k, \lambda) = 0 .
\end{align*}
Here 
\begin{align*}
    {d_{\mu\nu}}(k)
    = \delta_{\mu \nu} 
    - \frac{k_\mu k_\nu}{\abs{k}^2}
\end{align*}
denotes the transverse delta function. 
Since photons are massless, their dispersion relation is given by
\begin{align*}
    \omega(k) = \abs{k} .
\end{align*}
The boson Fock space $\mathcal{F}(\mathfrak{h})$ over $\mathfrak{h}$ is defined by the infinite direct sum of the symmetric $n$-fold tensor product of $\mathfrak{h}$:
\begin{align*}
    \mathcal{F}(\mathfrak{h})
    = \mathbb{C}
    \oplus \bigoplus_{n = 1}^\infty [\otimes_\mathrm{sym}^n \mathfrak{h}] .
\end{align*}
We denote by $\mathcal{F}(\mathfrak{h}) = \mathcal{F}$.
The Fock vacuum is defined by
\begin{align*}
    \Omega
    = 1 \oplus 0 \oplus 0 \oplus \cdots \in \mathcal{F} .
\end{align*}
The vector $\Phi$ belonging to $\mathcal{F}$ is denoted by $\Phi = {\{\Phi^{(n)}\}}_{n = 0}^\infty$, where $\Phi^{(n)}\in \otimes_\mathrm{sym}^n \mathfrak{h}$ for $n \geq 1$ and $\Phi^{(0)} \in \mathbb{C}$.
The creation operator acting in $\mathcal{F}$ is given by
\begin{align*}
    {({a^\dag}(f) \Phi)}^{(n)}
    &= \sqrt{n} {S_n}(f \otimes \Phi^{(n - 1)}) ,
    \quad
    {({a^\dag}(f) \Phi)}^{(0)} = 0 .
\end{align*}
Here $S_n : \otimes^n \mathfrak{h} \to \otimes_\mathrm{sym}^n \mathfrak{h}$ is the symmetrization operator.
The annihilation operator is defined by $a(f) = {({a^\dag}(\bar{f}))}^*$.
They satisfy the canonical commutation relations:
\begin{align*}
    [a(f), {a^\dag}(g)]
    &= {(\bar{f}, g)}_{\mathfrak{h}},
    \quad
    [a(f), a(g)]
    = [{a^\dag}(f), {a^\dag}(g)] = 0 .
\end{align*}
We introduce operator-valued distributions $a(k, \lambda)$ and ${a^\dag}(k, \lambda)$.
Then the annihilation and creation operators are formally written by
\begin{align*}
    a(f)
    = \sum_{\lambda = 1, 2}
        \int_{\mathbb{R}^3}
            {f(k, \lambda)} a(k, \lambda)
        \, \mathrm{d}k ,
    \quad {a^\dag}(f)
    = \sum_{\lambda = 1, 2}
        \int_{\mathbb{R}^3}
            f(k, \lambda) {a^\dag}(k, \lambda)
        \,\mathrm{d}k .
\end{align*}

\subsection{The Pauli-Fierz Hamiltonian}
The Hilbert space describing the state space of a single electron is
\begin{align*}
    \mathfrak{h}_{\mathrm{el}} = L^2(\mathbb{R}^3) .
\end{align*}
Therefore the total Hilbert space $\mathcal{H}$ of the Pauli-Fierz model is given by
\begin{align*}
    \mathcal{H} = \mathfrak{h}_{\mathrm{el}} \otimes \mathcal{F} .
\end{align*}
We introduce an ultraviolet cutoff function $\varphi \in {\mathscr S}'(\mathbb{R}^3)$.
Let $\hat{\varphi}$ be the Fourier transform of~$\varphi$. 
Throughout this paper we assume \cref{ass:uv} below.
\begin{assumption}\label{ass:uv} %
{\rm     We assume that $\hat{\varphi} \in L_{\mathrm{loc}}^1(\mathbb{R}^3)$, $\hat{\varphi}$ is rotation invariant, and
    \begin{align*}
        \overline{{\hat{\varphi}}(k)} = {\hat{\varphi}}(- k),
        \quad \sqrt{\omega} \hat{\varphi},
        \frac{\hat{\varphi}}{\sqrt{\omega}},
        \frac{\hat{\varphi}}{\omega} \in L^2(\mathbb{R}^3). 
    \end{align*}
}\end{assumption}
The first condition ensures that the cutoff function is real-valued, while the latter conditions ensure that the Pauli-Fierz Hamiltonian is self-adjoint.

A physical back ground of the introduction of the ultraviolet cutoff is as follows:  
If electrons are treated as point particles, they can couple to photons of arbitrarily short wavelengths.
Since a point particle is spatially localized at a single point, it includes all momentum components uniformly in momentum space.
In particular, the high-momentum components are not suppressed at all.
This means that the electron self-energy becomes infinite, and the vacuum expectation of the Hamiltonian becomes infinite.
However, real electrons have a finite size and are not point particles.
So, their interactions with photons having very short wavelengths should be suppressed.
The ultraviolet cutoff is introduced based on this physical consideration and serves as a fundamental assumption of the model.

For each $x \in \mathbb{R}^3$, we define the mode function by 
\begin{align*}
    {f^\nu_{x}}(k, \lambda) &= \frac{{\hat{\varphi}}(k)}{\sqrt{\omega(k)}} {e_{\nu}}(k, \lambda) e^{- i k \cdot x} ,
    \qquad \nu = 1, 2, 3 .
\end{align*}
We can see that $f^\nu_x \in L^2(\mathbb{R}^3 \times \{1,2\})$ holds under the condition ${\hat{\varphi}} / {\sqrt{\omega}} \in L^2(\mathbb{R}^3)$ for each $x \in \mathbb{R}^3$.
The quantized radiation field and its conjugate momentum are defined by
\begin{align*}
    {A_{\nu}}(x) &= \frac{1}{\sqrt{2}} \left({a^\dag}(f^\nu_{x}) +a(\bar f^\nu_{x})\right),  \\
    {{\Pi}_{\nu}(x)} &= \frac{i}{\sqrt{2}} \left({a^\dag}(f^\nu_{x}) -a(\bar f^\nu_{x}) \right), 
\end{align*}
respectively. Formal expressions in terms of $a(k, \lambda)$ and ${a^\dag}(k, \lambda)$ are given by
\begin{align*}
    {A_{\nu}}(x)
    &= \frac{1}{\sqrt{2}}
    \sum_{\lambda = 1, 2}
        \int_{\mathbb{R}^3}
            \frac{1}{\sqrt{\omega(k)}}
            \left( {\hat{\varphi}}(k) {a^\dag}(k, \lambda) e^{- i k \cdot x}+\overline{{\hat{\varphi}}(k)} a(k, \lambda) e^{i k \cdot x}  \right)
            {e_{\nu}}(k, \lambda)
        \, \mathrm{d}k , \\
    {{\Pi}_{\nu}}(x)
    &= \frac{i}{\sqrt{2}}
    \sum_{\lambda = 1, 2}
        \int_{\mathbb{R}^3}
            \frac{1}{\sqrt{\omega(k)}}
            \left({\hat{\varphi}}(k) {a^\dag}(k, \lambda) e^{- i k \cdot x}-\overline{{\hat{\varphi}}(k)} a(k, \lambda) e^{i k \cdot x} \right)
            {e_{\nu}}(k, \lambda)
        \, \mathrm{d}k .
\end{align*}
We write $A(x) = ({A_1}(x), {A_2}(x), {A_3}(x))$ and $\Pi(x) = ({{\Pi}_1}(x), {{\Pi}_2}(x), {{\Pi}_3}(x))$.
From these definitions and the transversality condition $k \cdot e(k, \lambda) = 0$, it follows that ${\nabla}_x \cdot A(x) = 0$.
We set
\begin{align*}
    {A_{\nu}} = \int_{\mathbb{R}^3}^{\oplus} {A_{\nu}}(x) \, \mathrm{d}x.
\end{align*}
The free field Hamiltonian is defined by the differential second quantization of $\omega$:
\begin{align*}
    H_\mathrm{rad} = \mathrm{d}\Gamma(\omega).
\end{align*}
Finally we explain external potential. 
Let ${(B_t)}_{t \geq 0}$ be 3D-Brownian motion on a probability space $(\mathcal{X}, \mathcal{F}, \mathcal{W})$ starting at $x$.
Let $\mathbb{E}^x$ be the expectation with respect to $\mathcal{W}$. For simplicity we $\mathbb{E}^0 = \mathbb{E}$ in what follows.
We decompose the measurable function $V: \mathbb{R}^3 \to \mathbb{R} $ into its positive and negative parts as $V = V_{+} - V_{-}$, where $V_{-} = \max\{- V, 0\}$ and $V_{+} = V + V_{-}$.
We assume that the negative part $V_{-}$ belongs to the Kato class, i.e.,
\begin{align*}
    \lim_{T \downarrow 0} \sup_{x \in \mathbb{R}^3} \mathbb{E}^x \left[ \int_0^T V_{-}(B_t) \, \mathrm{d}t \right] = 0
\end{align*}
and that the positive part $V_{+}$ belongs to the local Kato class.
Namely, for any compact set $K \subset \mathbb{R}^3$, the above condition holds with ${\mathbbm{1}}_{K} V_{+}$ put in place of $V_{-}$. The set of Kato class functions is denoted by $\mathcal{K}$ and the set of local Kato class functions by $\mathcal{K}_{\mathrm{loc}}$.
This condition means that the potential is sufficiently mild along Brownian trajectories in the short-time limit, so that it does not produce excessively singular effects on typical particle paths.

Under these assumptions, we define the Pauli-Fierz Hamiltonian by
\begin{align*}
    H_\alpha = \frac{1}{2} {(- i \nabla \otimes \mathbbm{1} - \alpha A)}^2 + V \otimes \mathbbm{1} + \mathbbm{1} \otimes H_\mathrm{rad}, 
\end{align*}
where $\alpha \in \mathbb{R}$ denotes the coupling constant.
Note that, in this definition, we have used the identification $\mathfrak{h} \cong \int_{\mathbb{R}^3}^{\oplus} \mathcal{F} \, \mathrm{d}x$.
In what follows, we simply write
\begin{align*}
    H_\alpha = \frac{1}{2} {(- i \nabla - \alpha A)}^2 + V + H_\mathrm{rad} .
\end{align*}

We conclude this section with a comment on the self-adjointness of $H_\alpha$.
If the external potential $V$ is relatively bounded with respect to $- (1 / 2) \Delta$ with relative bound strictly less than one, then $H_\alpha$ is self-adjoint on $D(- \Delta) \cap D(H_\mathrm{rad})$; see  \cite{hir00a,hir02b,HH08,mar15}. 
Moreover, if $V$ belongs to the Kato class, $H_\alpha$ admits a self-adjoint realization for which a Feynman-Kac formula for $e^{- T H_\alpha}$ can be established. We refer the reader to~\cite{HH10} for details.

\subsection{Feynman-Kac formula}
In this section, we derive a Feynman--Kac formula for $e^{-tH_\alpha}$.
For details, we refer the reader to~\cite[Chapter 3, Section 3]{HL26}.
To this end, we use the $Q$-space representation instead of the Fock representation introduced in the previous section.
Before introducing the $Q$-space representation, we describe the relation between $A_\nu$ and $\Pi_\nu$.
Let $N$ be the number operator on $\mathcal{F}$.
Then
\begin{align*}
    i[N,A_\nu(x)] = \Pi_\nu(x).
\end{align*}
Moreover,
\begin{align*}
    e^{-i\frac{\pi}{2}N} A_\nu(x) e^{i\frac{\pi}{2}N}
    = -\Pi_\nu(x).
\end{align*}
See \cref{a1} in Appendix for details.

Two Gaussian configuration spaces are introduced: the Minkowskian configuration space $Q$ and the Euclidean configuration space $Q_{\mathrm{E}}$.
On the Minkowskian configuration space $Q$, we define a probability space $(Q,\Sigma,\mu)$ carrying a centered Gaussian random variable ${\{\hat A(f)\}}_{f \in \oplus^3 L^2(\mathbb{R}^3)}$ with covariance
\begin{align*}
    \mathbb{E}_{\mu}[{\hat{A}}(f) {\hat{A}}(g)]
    = \frac{1}{2}
    \int_{\mathbb{R}^3}
        \overline{{\hat{f}_{\mu}}(k)}
        \left({\delta}_{\mu \nu} - \frac{k_\mu k_\nu}{\abs{k}^2} \right)
        {\hat{g}_{\nu}}(k)
    \, \mathrm{d}k .
\end{align*}
Here $\hat f$ denotes the Fourier transform of $f$:
\begin{align*}
    {\hat{f}}(k) = \frac{1}{{(2 \pi)}^{3 / 2}} \int_{\mathbb{R}^3} e^{- i k \cdot x} f(x) \, \mathrm{d}x .
\end{align*}
On the other hand, on the Euclidean configuration space $Q_{\mathrm{E}}$, we define a probability space $(Q_{\mathrm{E}}, {\Sigma}_{\mathrm{E}}, {\mu}_{\mathrm{E}})$ carrying a centered Gaussian random variable ${\{\hat{A}_{\mathrm{E}}(F)\}}_{F \in \oplus^3 L^2(\mathbb{R}^4)}$ with covariance
\begin{align*}
    \mathbb{E}_{{\mu}_{\mathrm{E}}}[\hat{A}_{\mathrm{E}}(F) \hat{A}_{\mathrm{E}}(G)]
    &= \frac{1}{2}
    \int_{\mathbb{R}^4}
        \overline{{\hat{F}_{\mu}}(k_0, k)}
        \left({\delta}_{\mu \nu} - \frac{k_\mu k_\nu}{\abs{k}^2} \right)
        {\hat{G}_{\nu}}(k_0, k)
    \,\mathrm{d}k_0 \,\mathrm{d}k.
\end{align*}

The Minkowskian and Euclidean configuration spaces are connected by a family of isometries ${\{\mathrm{j}_t\}}_{t\in\mathbb{R}}$.
Let
$\mathrm{j}_t:\oplus^3L^2(\mathbb{R}^3)\to\oplus^3L^2(\mathbb{R}^4)$
be the isometry defined by
\begin{align*}
    \widehat{\mathrm{j}_t f}(k_0, k)
    = \frac{1}{\sqrt{\pi}}
    \sqrt{\frac{\omega(k)}{\abs{k_0}^2 + {{\omega}(k)}^2}}
    {\hat{f}}(k)
    e^{- i t k_0}.
\end{align*}
It holds that 
\begin{align}\label{eq:jt} %
    \mathrm{j}_t^* \mathrm{j}_s = e^{- \abs{t - s} \hat{\omega}} ,
    \qquad t, s \in \mathbb{R},
\end{align}
where \[\hat{\omega} = \omega(-i \nabla_k) = \sqrt{- \Delta_k}.\]
Let $\hat H_\mathrm{rad} $ be given by 
\begin{align*}
    \hat{H}_{\mathrm{f}} : \hat{A}(f_1) \cdots \hat{A}(f_n): \: = \: 
    : \hat{A}(\hat{\omega} f_1) \cdots \hat{A}(\hat{\omega} f_n) :, 
\end{align*}
where $:\quad:\,$ denotes the Wick product~\cite[\text{Section 1.2.6}]{HL26}.
Moreover let $\hat N$ and $\hat N_{\mathrm{E}}$ be number operators on 
$L^2(Q)$ and $L^2(Q_{\mathrm{E}})$, respectively, i.e., 
\begin{align*}
    &\hat{N} : {\hat{A}}(f_1) \cdots {\hat{A}}(f_n): \: = \: n : {\hat{A}}(f_1) \cdots {\hat{A}}(f_n): , \\
    & \hat{N}_{\mathrm{E}} : {\hat{A}_{\mathrm{E}}}(f_1) \cdots {\hat{A}_{\mathrm{E}}}(f_n): \: = \: n : {\hat{A}_{\mathrm{E}}}(f_1) \cdots {\hat{A}_{\mathrm{E}}}(f_n): .
\end{align*}
We define the conjugate momentum of ${\hat{A}}(f)$ and ${\hat{A}_{\mathrm{E}}}(f)$ by
\begin{align*}
    \hat{\Pi}(f) = i [\hat{N}, {\hat{A}}(f)] ,
    \quad {\hat{\Pi}_{\mathrm{E}}}(f) = i [\hat{N}_{\mathrm{E}}, {\hat{A}_{\mathrm{E}}}(f)] ,
\end{align*}
respectively. 
Let $S$ be a contraction operator from 
$\oplus^3 L^2(\mathbb{R}^3)$ to $\oplus^3 L^2(\mathbb{R}^4)$.
Then the second quantization ${{\Gamma}_{\mathrm{int}}}(S) : L^2(Q) \to L^2(Q_{\mathrm{E}})$ is defined by
\begin{align*}
    {{\Gamma}_{\mathrm{int}}}(S) : {\hat{A}}(f_1) \cdots {\hat{A}}(f_n) \: = \: : {\hat{A}_{\mathrm{E}}}(S f_1) \cdots {\hat{A}_{\mathrm{E}}}(S f_k):
\end{align*}
and ${{\Gamma}_{\mathrm{int}}}(S) \mathbbm{1} = \mathbbm{1}$.
In particular, the operator $\mathrm{J}_t = {{\Gamma}_{\mathrm{int}}}(\mathrm{j}_t)$ is an isometry from $L^2(Q)$ to $L^2(Q_{\mathrm{E}})$ and by \cref{eq:jt} it satisfies that
\begin{align*}
    \mathrm{J}_t^* \mathrm{J}_s = e^{- \abs{t - s} \hat{H}_{\mathrm{f}}} ,
    \qquad t, s \in \mathbb{R} .
\end{align*}
We introduce two phase rotations on $L^2(Q)$ and $L^2(Q_{\mathrm{E}})$:
\begin{align*}
    \hat{\Theta} = e^{i \frac{\pi}{2} \hat{N}},
    \qquad \hat{\Theta}_{\mathrm{E}} = e^{i \frac{\pi}{2} \hat{N}_{\mathrm{E}}},
\end{align*}
respectively. 
The phase rotations $\hat{\Theta}$ and $\hat{\Theta}_{\mathrm{E}}$ satisfy
\begin{align*}
    \hat{\Theta}^{-1} \hat{A}(f) \hat{\Theta} = -{\hat{\Pi}}(f) ,
    \qquad {\hat{\Theta}_{\mathrm{E}}}^{-1} {\hat{A}_{\mathrm{E}}}(F) \hat{\Theta}_{\mathrm{E}} = -{\hat{\Pi}_{\mathrm{E}}}(F) ,
\end{align*}
and $\mathrm{J}_t$ intertwines $\hat{\Theta}$ and $\hat{\Theta}_{\mathrm{E}}$ in the sense that
\begin{align*}
    \mathrm{J}_t \hat{\Theta} = \hat{\Theta}_{\mathrm{E}} \mathrm{J}_t ,
    \qquad t \in \mathbb{R}.
\end{align*}
The phase rotation may be regarded as the analogue of the Fourier transform on the boson Fock space.
Indeed, it satisfies
\begin{align*}
    {\hat \Theta}^4 =\mathbbm{1} ,
\end{align*}
just as the Fourier transform on $L^2(\mathbb{R}^3)$.
The boson Fock space $\mathcal{F}$ is unitarily equivalent to $L^2(Q)$ which is implemented by the so-called Wiener-Ito-Segal isomorphism:
$\mathrm{U} : \mathcal{F} \to L^2(Q)$ given by 
\begin{align*}
    &\mathrm{U} \Omega = \mathbbm{1},\\
    &\mathrm{U} : A(f_1) \cdots A(f_n) : \Omega = : \hat{A}(f_1) \cdots \hat{A}(f_n) :
\end{align*}
As is seen above in this representation, the field operators $A(x)$ are realized as Gaussian random variables.
This formulation is particularly convenient for deriving functional integral representations of the semigroup generated by $H_\alpha$. 
Let
\begin{align*}
    \tilde{\varphi} = {\left(\frac{\hat{\varphi}}{\sqrt{\omega}}\right)}^{\vee}.
\end{align*}
Under the Wiener - Ito - Segal isomorphism $\mathrm{U}$, operators on $\mathcal{F}$ are transformed to operators on $L^2(Q)$.
We can see that 
\begin{align*}
    \hat A(\tilde{\varphi}(\cdot - x)) &= \mathrm{U} A(x) {\mathrm{U}}^{-1} ,
    \quad \hat{H_\mathrm{rad}} = \mathrm{U} H_\mathrm{rad} {\mathrm{U}}^{-1} ,
    \quad \hat{N} = \mathrm{U} N {\mathrm{U}}^{-1} .
\end{align*}
The Pauli-Fierz Hamiltonian $H_\alpha $ is transformed to
\begin{align*}
    \hat H_\alpha = \mathrm{U} H_\alpha  \mathrm{U}^{-1}
    = \frac{1}{2}  {\left(- i \nabla - \alpha \hat A(\tilde{\varphi}(\cdot - x))\right)}^2 + V + \hat{H_\mathrm{rad}}.
\end{align*}
Note that  $\hat{H}_{\alpha}$ is self-adjoint on the Hilbert space $L^2(\mathbb{R}^3 \times Q)$. 
To distinguish operators in the $Q$-space representation from those in the Fock representation, we temporarily used the symbol like $\hat{H}_{\alpha} $ so far.
If no confusion arises, we omit the hat in what follows.
Moreover in what follows we denote $L^2(Q)$ by $\mathcal{F}$, $L^2(\mathbb{R}^3 \times Q)$ by $\mathcal{H}$, and $L^2(Q_{\mathrm{E}})$ by $\mathcal{E}$.
Hereafter, we omit the notations~$\hat{\quad}$ unless there is a particular risk of confusion.

Let $E_\alpha $ denote the infimum of the spectrum of $H_\alpha $.~i.e.
\begin{align*}
    H_\alpha {\varphi_\mathrm{g}}= E_\alpha \varphi_\mathrm{g}.
\end{align*}

\begin{lemma}[\cite{HH10} Theorem 2.18]
    Let $V$ be Kato-decomposable.
    For any $\Phi, \Psi \in \mathcal{H} $, the Feynman-Kac formula below holds:
    \begin{align}\label{eq:FKF} %
        {(\Phi, e^{- T H_\alpha} \Psi)}_\mathcal{H}
        &= \int_{\mathbb{R}^3}
            \mathbb{E}^{x} \left[
                e^{- \int_0^T V(B_t) \,\mathrm{d} t} 
                {(\Phi(B_0), \mathrm{J}_0^* e^{-i \alpha A_\mathrm{E}(K_T)} \mathrm{J}_T \Psi(B_T))}_\mathcal{F}
            \right] 
        \, \mathrm{d}x ,
    \end{align}
where
\begin{align*}
    K_T = \bigoplus_{j = 1}^3 \int_0^T \mathrm{j}_t \tilde{\varphi}(\cdot - B_t) \, \mathrm{d}B_t^j
\end{align*} 
denotes $\bigoplus_{j = 1}^3 L^2(\mathbb{R}^3)$-valued stochastic integral,  and 
    \begin{align*}
        (e^{- T H_\alpha } \Psi)(x)
        &= \mathbb{E}^{x} \left[
            e^{- \int_0^T V(B_t) \, \mathrm{d}t}
            \mathrm{J}_0^* e^{- i \alpha \hat{A}_{\mathrm{E}}(K_T)} \mathrm{J}_T \Psi(B_T)
        \right]
        \qquad {\mathrm{a.e. }} x \in \mathbb{R}^3 .
    \end{align*}
\end{lemma}
We also refer the readers to \cite{spo86,hir07, hir14, mat16, mat17, GMM17}. 
In this representation, the motion of the electron is described by a Brownian path, while the quantized radiation field appears as a Gaussian random variable.
Consequently, the matrix elements of the semigroup $e^{- t H_\alpha}$ can be computed as expectations with respect to the joint fluctuations of the electron's random trajectory and the Gaussian fluctuations of the photon field.

\subsection{Upper bounds of exponential decay}
We introduce two classes of external potentials.
\begin{definition}
\normalfont
The potential classes $\mathcal{V}_{\mathrm{upper}}$ and
$\mathcal{V}_{\mathrm{lower}}$ are defined as follows.
\begin{enumerate}[label=(\arabic*)]
    \item
    $V\in\mathcal{V}_{\mathrm{upper}}$ if and only if
    $V$ admits a decomposition $V=W-U$ satisfying:
    \begin{enumerate}[label=(\alph*)]
        \item
        $U\geq 0$ and $U\in L^p(\mathbb{R}^3)$
        for some $p\in(3/2,\infty)$;
        \item
        $W\in L^1_{\mathrm{loc}}(\mathbb{R}^3)$ and
$            W_\infty=\inf_{x\in\mathbb{R}^3}W(x)>-\infty$.
    \end{enumerate}
    \item
    $V\in\mathcal{V}_{\mathrm{lower}}$ if and only if
    $V$ admits a decomposition $V=W-U$ satisfying:
    \begin{enumerate}[label=(\alph*)]
        \item
        $U\geq 0$ and $U\in L^p(\mathbb{R}^3)$
        for some $p\in(3/2,\infty)$;
        \item
        $W\geq 0$ and
        $W\in\mathcal{K}_{\mathrm{loc}}(\mathbb{R}^3)$.
    \end{enumerate}
\end{enumerate}
\end{definition}
For $a>0$, define
\[
     \tilde W_a(x)=\inf\{W(y)\mid |y-x|<a\}.
\]

\begin{lemma}[
    {\cite[Lemma 3.2 and Theorem 3.3]{HH10};
     \cite[Lemma 3.54 and Corollary 3.55]{HL26}}
]\label{upper}
Let $V=W-U\in\mathcal{V}_{\mathrm{upper}}$.
For each $b\in(0,1/2)$, there exist constants
$D_1,D_2,D_3>0$, independent of $T$, $a$, and $x$,
such that, for all $T,a>0$,
\[
    \|\varphi_{\mathrm g}(x)\|_{\mathcal F}
    \leq
    D_1T^{-3/2}
    e^{(D_2\|U\|_{L^p}+E_\alpha)T}
    \left(
        D_3e^{-ba^2/(4T)}e^{-TW_\infty}
        +e^{-T\tilde W_a(x)}
    \right).
\]
\end{lemma}

\begin{corollary}\label{cupper}
Suppose that, for some real number $n>0$ and constants
$0<A\leq B$ and $c_1,c_2>0$,
\[
    A^2|x|^{2n}-c_1
    \leq V(x)
    \leq B^2|x|^{2n}+c_2,
    \qquad x\in\mathbb{R}^3.
\]
Let $0<\epsilon<\tfrac12$. 
Then there exist constants $C_\epsilon>0$ such that
\[
    \|\varphi_{\mathrm g}(x)\|_{\mathcal F}
    \leq
    C_\epsilon e^{-(b-\epsilon)\frac{A}{2^{n+4}}|x|^{n+1}},
    \qquad x\in\mathbb{R}^3.
\]
\end{corollary}

\begin{proof}
Set $W=V_+=\max\{V,0\}$ and $U=V_-=\max\{-V,0\}$.
Then $V=W-U$ and $W,U\geq0$.
By the lower bound on $V$,
\[
    0\leq U(x)\leq c_1\mathbbm{1}_K(x),
    \qquad
    K=
    \left\{
        x\in\mathbb{R}^3\mid 
        |x|\leq
        \left(\frac{c_1}{A^2}\right)^{1/(2n)}
    \right\}.
\]
Thus $U\in L^p(\mathbb{R}^3)$ for every $p\in[1,\infty)$.
Moreover,
\[
    0\leq W(x)\leq B^2|x|^{2n}+c_2,
\]
so $W\in L^1_{\mathrm{loc}}(\mathbb{R}^3)$ and
$W_\infty\geq0$.
Hence $V=W-U\in\mathcal{V}_{\mathrm{upper}}$.
Fix $p\in(3/2,\infty)$ and $b\in(0,1/2)$.
For $r=|x|\geq1$, set
\[
    a=\frac r2,
    \qquad
    T=\frac{2^n}{A}r^{1-n}.
\]
If $|y-x|<a$, then $|y|>r/2$.
Since $W\geq V$, we obtain
\[
    \tilde W_a(x)
    =\inf_{|y-x|<a}W(y)
    \geq \frac{A^2}{2^{2n}}r^{2n}-c_1.
\]
Consequently,
\[
    T\tilde W_a(x)
    \geq \frac{A}{2^n}r^{n+1}-c_1T,
    \qquad
    \frac{a^2}{4T}
    =\frac{A}{16\cdot2^n}r^{n+1}.
\]
Applying Lemma~\ref{upper} and using $W_\infty\geq0$,
we obtain
\begin{align*}
    \|\varphi_{\mathrm g}(x)\|_{\mathcal F}
    \leq
    D_1\left(\frac{A}{2^n}\right)^{3/2}
    r^{3(n-1)/2}
    e^{(D_2\|U\|_{L^p}+|E_\alpha|)T}
        \left(
        D_3e^{-\frac{b}{16}\frac{A}{2^n}r^{n+1}}
        +e^{c_1T}e^{-\frac{A}{2^n}r^{n+1}}
    \right).
\end{align*}
Set $M=D_2\|U\|_{L^p}+|E_\alpha|+c_1$.
Since $b/16<1$, it follows that
\[
    \|\varphi_{\mathrm g}(x)\|_{\mathcal F}
    \leq
    D_1(D_3+1)\left(\frac{A}{2^n}\right)^{3/2}
    r^{3(n-1)/2}
    \exp\left\{
        \frac{2^nM}{A}r^{1-n}
        -\frac{b}{16}\frac{A}{2^n}r^{n+1}
    \right\}.
\]
Choose $\kappa\in(0,b/16)$.
Since $n>0$,
\[
    r^{1-n}=o(r^{n+1}),
    \qquad
    \log r=o(r^{n+1})
    \quad\text{as }r\to\infty.
\]
Therefore,
\[
    \sup_{r\geq1}
    r^{3(n-1)/2}
    \exp\left\{
        \frac{2^nM}{A}r^{1-n}
        -\frac{\epsilon}{16}
         \frac{A}{2^n}r^{n+1}
    \right\}
    <\infty.
\]
Thus there exists a constant $C_\epsilon>0$, independent of $x$,
such that
\[
    \|\varphi_{\mathrm g}(x)\|_{\mathcal F}
    \leq C_\epsilon e^{-(b-\epsilon)\frac{A}{2^{n+4}}|x|^{n+1}},
    \qquad |x|\geq1.
\]
Finally, by continuity,
\[
    \sup_{|x|\leq1}
    \|\varphi_{\mathrm g}(x)\|_{\mathcal F}
    e^{(b-\epsilon)\frac{A}{2^{n+4}}|x|^{n+1}}
    <\infty.
\]
After increasing $C_\epsilon$, the estimate therefore holds
for all $x\in\mathbb{R}^3$.
This completes the proof.
\end{proof}

There is no doubt that the evaluation from above is extremely important, but this evaluation is weaker than ideal spatial decay.
In Agmon theory, the spatial decay is expected as 
$e^{- (1 \pm \varepsilon) A\frac{\abs{x}^{n + 1}}{n + 1}} \lbrace {}^{+ \varepsilon \; \mathrm{lower}}_{- \varepsilon \; \mathrm{upper}}$.
What this implies is that losses appear to be occurring in the evaluation from above.
It is merely observing large deviations and does not follow the typical optimal path seen in Agmon theory.

\section{Lower bounds}\label{sec:3} %
\subsection{Phase rotations}
Since the norm of a vector dominates the norm of its vacuum component, it holds that
\begin{align*}
    \norm{{\varphi_\mathrm{g}}(x) }_\mathcal{F}
    \geq \norm*{\varphi_\mathrm{g}^{(0)}(x)}_{\mathbb{C}}
    = \abs{{(\mathbbm{1}, {\varphi_\mathrm{g}}(x))}_\mathcal{F}} .
\end{align*}
The subsequent lemmas will be used to derive a lower bound on
$\abs{{(\mathbbm{1}, {\varphi_\mathrm{g}}(x))}_\mathcal{F}}$. 
\begin{lemma}
It follows that 
    \begin{align}\label{eq:lem1} %
        {(\mathbbm{1}, {\varphi_\mathrm{g}}(x))}_\mathcal{F}
        = \mathbb{E}^{x}\left[
            e^{- \int_0^T (V(B_t) - E_\alpha) \, \mathrm{d}t}
            {(\mathbbm{1}, e^{- i \alpha A_\mathrm{E}(K_T)} \mathrm{J}_T {\varphi_\mathrm{g}}(B_T))}_{\mathcal{E}}
        \right]
        \qquad {\mathrm{a.e. }} x \in \mathbb{R}^3 .
    \end{align}
\end{lemma}
\begin{proof} 
    By Feynman-Kac formula \cref{eq:FKF}, we have
    \begin{align}\label{eq:lem1.1} %
        {(g \otimes \mathbbm{1}, \varphi_\mathrm{g})}_{\mathcal{H} }
        &= 
        {(g \otimes \mathbbm{1}, e^{- T (H_\alpha  - E_\alpha )} \varphi_\mathrm{g})}_{\mathcal{H} } \notag \\
        &= \int_{\mathbb{R}^3}
            \mathbb{E}^{x}\left[
                e^{- \int_0^T (V(B_t) - E_\alpha) \,\mathrm{d}t}
                {(\mathrm{J}_0 g(x) \mathbbm{1}, e^{- i \alpha A_\mathrm{E}(K_T)} \mathrm{J}_T {\varphi_\mathrm{g}}(B_T))}_{\mathcal{E}}
            \right]
        \, \mathrm{d}x \notag \\
        &= \int_{\mathbb{R}^3}
            {\bar{g}}(x)
            \mathbb{E}^{x}\left[
                e^{- \int_0^T (V(B_t) - E_\alpha) \, \mathrm{d}t}
                {(\mathbbm{1}, e^{- i \alpha A_\mathrm{E}(K_T)} \mathrm{J}_T {\varphi_\mathrm{g}}(B_T))}_{\mathcal{E}}
            \right]
        \, \mathrm{d}x .
    \end{align}
    On the other hand,
    \begin{align}\label{eq:lem1.2} %
        {(g \otimes \mathbbm{1}, \varphi_\mathrm{g})}_{\mathcal{H} }
        &= \int_{\mathbb{R}^3}
            {\bar{g}}(x)
            {(\mathbbm{1}, {\varphi_\mathrm{g}}(x))}_\mathcal{F}
        \, \mathrm{d}x .
    \end{align}
    Comparing \cref{eq:lem1.1} and \cref{eq:lem1.2}, we obtain \cref{eq:lem1}.
\end{proof}

In the integrand of \cref{eq:lem1}, there appears the oscillatory factor $e^{-i\alpha A_\mathrm{E}(K_T)}$.
Since $A_\mathrm{E}(K_T)\in\mathbb{R}$, the quantity $e^{-i\alpha A_\mathrm{E}(K_T)}$ is complex-valued.
This makes it difficult to derive lower bounds from \cref{eq:lem1} directly.

\begin{lemma}\label{lem:2} %
    The ground state $\varphi_\mathrm{g}$ of $H_\alpha$ satisfies that 
    $x \mapsto {(\mathbbm{1}, {\varphi_\mathrm{g}}(x))}_\mathcal{F}$ has the continuous version and the continuous version is strictly positive:
    \begin{align*}
        {(\mathbbm{1}, {\varphi_\mathrm{g}}(x))}_\mathcal{F} > 0,
        \qquad x \in \mathbb{R}^3. 
    \end{align*}
\end{lemma}
\begin{proof}
    The continuity of ${(\mathbbm{1}, {\varphi_\mathrm{g}}(\cdot))}_\mathcal{F}$ is proven in \cref{sec:a2} in Appendix and the positivity by~\cite[\text{Corollary 3.12}]{hir00b}.
\end{proof}

Let us define the continuous function $m(\cdot)$ by
\begin{align*}
    m(y)
    = {(\mathbbm{1}, \mathrm{J}_T \Theta^{-1} {\varphi_\mathrm{g}}(y))}_{\mathcal{E}}
    = {(\mathbbm{1}, \Theta^{-1} {\varphi_\mathrm{g}}(y))}_\mathcal{F}
    = {(\mathbbm{1}, {\varphi_\mathrm{g}}(y))}_\mathcal{F} .
\end{align*}
The following lemma provides the key observation that overcomes this difficulty.
 \begin{lemma}[Key lemma I]\label{lem:3} %
    We have
    \begin{align}\label{eq:lem3} %
        &{(\mathbbm{1}, e^{-i\alpha A_\mathrm{E}(K_T)} \mathrm{J}_T {\varphi_\mathrm{g}}(B_T))}_{\mathcal{E}} \notag \\
        &\geq m(B_T)
        \exp(- \alpha^2 \mathbb{E}_{\mu_{\mathrm{E}}}[{A_\mathrm{E}}(K_T) {A_\mathrm{E}}(K_T)])
        \exp(- \alpha \frac{{(A_\mathrm{E}(K_T)\mathbbm{1}, \mathrm{J}_T \Theta^{-1} {\varphi_\mathrm{g}}(B_T))}_{\mathcal{E}}}{m(B_T)})
    \end{align}
\end{lemma}
\begin{proof}
    From the properties of the phase rotation $\Theta$, it follows that
    \begin{align*}
        {(\mathbbm{1}, e^{- i \alpha  A_\mathrm{E}(K_T)} \mathrm{J}_T {\varphi_\mathrm{g}}(B_T))}_{\mathcal{E}}
        &= {({\Theta_\mathrm{E}}^{-1} \mathbbm{1}, {\Theta_\mathrm{E}}^{-1} e^{- i \alpha {A_\mathrm{E}}(K_T)} \mathrm{J}_T {\varphi_\mathrm{g}}(B_T))}_{\mathcal{E}}  \\
        &= {(\mathbbm{1}, e^{+i\alpha {\Pi}_{\mathrm{E}}(K_T)} \Theta_\mathrm{E}^{-1} \mathrm{J}_T {\varphi_\mathrm{g}}(B_T))}_{\mathcal{E}}  \\
        &= {(e^{-i \alpha {\Pi}_{\mathrm{E}}(K_T)} \mathbbm{1}, \Theta_\mathrm{E}^{-1} \mathrm{J}_T {\varphi_\mathrm{g}}(B_T))}_{\mathcal{E}}  \\
        &= {e^{- \alpha^2\mathbb{E}_{\mu_{\mathrm{E}}}[A_\mathrm{E}(K_T) A_\mathrm{E}(K_T)]} (e^{-\alpha A_\mathrm{E}(K_T)}\mathbbm{1},  \mathrm{J}_T \Theta^{-1} {\varphi_\mathrm{g}}(B_T))}_{\mathcal{E}} . 
    \end{align*}
    Furthermore,
    \begin{align*}
        {(e^{- \alpha {A_\mathrm{E}}(K_T)}\mathbbm{1}, \mathrm{J}_T \Theta^{-1} {\varphi_\mathrm{g}}(B_T))}_{\mathcal{E}}  
        = m(B_T) \frac{{(e^{- \alpha {A_\mathrm{E}}(K_T)}\mathbbm{1}, \mathrm{J}_T \Theta^{-1} {\varphi_\mathrm{g}}(B_T))}_{\mathcal{E}}}{m(B_T)}
    \end{align*}
    holds, and by Jensen's inequality,
    \begin{align*}
        \frac{{(e^{- \alpha A_\mathrm{E}(K_T)}\mathbbm{1}, \mathrm{J}_T \Theta^{-1} {\varphi_\mathrm{g}}(B_T))}_{\mathcal{E}}}{m(B_T)}
        &\geq \exp({- \alpha \frac{{(A_\mathrm{E}(K_T)\mathbbm{1}, \mathrm{J}_T \Theta^{-1} {\varphi_\mathrm{g}}(B_T))}_{\mathcal{E}}}{m(B_T)}})
    \end{align*}
    holds true.
\end{proof}

To estimate the right-hand side of \cref{eq:lem3} from below, it suffices to derive an upper bound on the exponent ${({A_\mathrm{E}}(K_T), \mathrm{J}_T \Theta^{-1} {\varphi_\mathrm{g}}(B_T))}_{\mathcal{E}}$ and a lower bound of $m(B_T)$.
\begin{lemma}\label{lem:4} %
    It follows that
    \begin{align*}
        \sup_{x \in \mathbb{R}^3} \norm{ {\varphi_\mathrm{g}}(x) }_\mathcal{F}
        \leq {(2 \pi T)}^{- 3 / 4}
        \sup_{x \in \mathbb{R}^3}
        {\mathbb{E}\left[e^{- 2 \int_0^T (V(B_s + x) - E_\alpha) \,\mathrm{d} s}\right]}^{1 / 2} \norm{\varphi_\mathrm{g}}_\mathcal{H} .
    \end{align*}
\end{lemma}
\begin{proof}
    By the Feynman-Kac formula we have
    \begin{align*}
        \norm{{\varphi_\mathrm{g}}(x) }_\mathcal{F}
        = \norm*{e^{- T (H_\alpha - E_\alpha)}{\varphi_\mathrm{g}}(x)}_\mathcal{F}
        \leq \mathbb{E}\left[
            e^{- \int_0^T (V(B_s + x) - E_\alpha) \, \mathrm{d}s}
            \norm{ {\varphi_\mathrm{g}}(B_T + x)}_\mathcal{F}
        \right] .
    \end{align*}
    By the Schwarz inequality we have 
    \begin{align*}
        \norm{ {\varphi_\mathrm{g}}(x) }_\mathcal{F}
        &\leq {\mathbb{E}\left[e^{- 2 \int_0^T (V(B_s + x) - E_\alpha) \, \mathrm{d}s}\right]}^{1 / 2}
        {\mathbb{E}\left[\norm{ {\varphi_\mathrm{g}}(B_T+x) }_\mathcal{F}^2\right]}^{1 / 2}\\
        &\leq
        {\mathbb{E}\left[e^{- 2 \int_0^T (V(B_s + x) - E_\alpha) \, \mathrm{d}s}\right]}^{1/2}
        {\left(
            \int_{\mathbb{R}^3}
                {(2 \pi T)}^{- 3 / 2}
                \norm{ {\varphi_\mathrm{g}}(y + x) }_\mathcal{F}^2
                e^{- \abs{y}^2 / (2 T)}
            \, \mathrm{d}y
        \right)}^{1 / 2}.
    \end{align*}
    Then the lemma follows. 
\end{proof}

Since $V$ is Kato-decomposable, 
we have $\sup_{x \in \mathbb{R}^3} \mathbb{E}[e^{- 2 \int_0^T (V(B_s + x) - E_\alpha) \, \mathrm{d}s}] < + \infty$.   
Let 
\begin{align*}
    g_{\infty} = \sup_{x \in \mathbb{R}^3} \norm{ {\varphi_\mathrm{g}}(x) }_\mathcal{F} .
\end{align*}
\begin{lemma}
    It follows that 
    \begin{align*}
        {(\mathbbm{1}, {\varphi_\mathrm{g}}(x))}_\mathcal{F}
        \geq \mathbb{E}\left[
            e^{- \int_0^T (V(B_t + x) - E_\alpha) \, \mathrm{d}t}
            e^{- \tilde{X}}
            m(B_T + x)
        \right], 
    \end{align*}
    where 
    \begin{align*}
        \tilde{X}
        = \frac{\alpha^2}{2} \norm{ K_T }_{\oplus^3 L^2}^2 
        + \frac{\abs{\alpha} g_{\infty}}{m(B_T)}{\norm{ K_T }_{\oplus^3 L^2}} .
    \end{align*}
\end{lemma}
\begin{proof}
    We have 
    \begin{align*}
        &\alpha^2 \mathbb{E}_{{\mu}_{\mathrm{E}}}[{A_\mathrm{E}}(K_T) {A_\mathrm{E}}(K_T)]
        \leq \frac{\alpha^2}{2} \norm{ {K}_T }_{\oplus^3 L^2}^2,\\
        &\alpha {({A_\mathrm{E}}(K_T)\mathbbm{1}, \mathrm{J}_T \Theta^{-1} {\varphi_\mathrm{g}}(B_T))}_{\mathcal{E}}
        \leq \abs{\alpha} \norm{ {A_\mathrm{E}}(K_T) }_{\mathcal{E}} \norm{ {\varphi_\mathrm{g}}(B_T) }_\mathcal{F} .
    \end{align*}
    Here$\norm{ {A_\mathrm{E}}(K_T)\mathbbm{1} }_{\mathcal{E}}^2 = \frac{1}{2} \norm{ K_T }_{\oplus^3 L^2}^2$ and, by \cref{eq:lem3} and \cref{lem:4} we obtain the lemma.
\end{proof}

\subsection{Estimates of exponential moments of double stochastic integrals}
Let
\begin{align*}
    S_T = \{\abs{B_t} \leq 1, 0 \leq t \leq T\} .
\end{align*}
Then
\begin{align*}
    {(\mathbbm{1}, {\varphi_\mathrm{g}}(x))}_\mathcal{F} 
    &\geq \mathbb{E}\left[
        e^{- \int_0^T (V(B_t + x) - E_\alpha) \, \mathrm{d}t}
        e^{- \tilde{X}}
        m(B_T + x)
    \right] \\
    &\geq \mathbb{E}\left[
        \mathbbm{1}_{S_T}
        e^{- \int_0^T (V(B_t + x) - E_\alpha) \, \mathrm{d}t}
        e^{- \tilde{X}}
        m(B_T + x)
    \right] ,
\end{align*}
where $\mathbbm{1}_{S_T}$ is the indicator function of $S_T$.
Let
\begin{align}
    m = \inf_{y \in B_0} m(y).
\end{align}
Since the unit ball $B_0 \subset \mathbb{R}^3$ is compact, $m$ is strictly positive and we obtain that
    \begin{align*}
        \tilde{X} \leq X
    \end{align*}
on $S_T$, where 
\begin{align}
    X
    = \frac{\alpha^2}{2}  
    \norm{ {K}_T }_{\oplus^3 L^2}^2
    + \frac{\abs{\alpha} g_\infty}{m}
    \norm{ {K}_T }_{\oplus^3 L^2} .
\end{align}
Hence we obtain that
\begin{align*}
    {(\mathbbm{1}, {\varphi_\mathrm{g}}(x))}_\mathcal{F}
    \geq \mathbb{E}\left[
        \mathbbm{1}_{S_T}
        e^{- \int_0^T (V(B_t + x) - E_\alpha) \, \mathrm{d}t}
        e^{- X}
        m(B_T + x)
    \right] .
\end{align*}
We need to estimate $\mathbb{E}[e^{a X}]$ from above. 
{Formally} $\norm{{K}_T}_{\oplus^3 L^2}^2$ is of the form:
\begin{align*}
    \norm{{K}_T}_{\oplus^3 L^2}^2
    = \sum_{\mu, \nu = 1}^3
        \int_0^T
            \int_0^T \left(
                \int_{\mathbb{R}^3}
                    {d_{\mu \nu}}(k)
                    \frac{\abs{\hat{\varphi}(k)}^2}{\omega(k)}
                    e^{- i k \cdot(B_s - B_t)}
                    e^{- \abs{s - t} \omega(k)}
                \, \mathrm{d}k
            \right) \, \mathrm{d}B_s^\mu
        \,\mathrm{d} B_t^\nu.
\end{align*} 
Although this expression is often referred to as a stochastic double integral, this terminology is only formal.
Indeed, the integrand does not satisfy the adaptedness required for stochastic integrals. Consequently, estimating $\mathbb{E}[e^{a X}]$ is far from straightforward.
See e.g.~\cite{spo86,BH09}.
Ideally, one would like to reduce the stochastic integral to an ordinary Riemann integral by applying Ito formula or related techniques.
However, to the best of our knowledge, such a reduction is not available for the Pauli-Fierz model we study in this paper.
Instead, we derive the desired estimate by combining exponential martingale techniques with the Jensen inequality and several auxiliary estimates.
Let
\begin{align*}
    \triangle = 24 \alpha_*^2 \norm{ \frac{\hat{\varphi}}{\sqrt{\omega}} }^2 ,
    \quad \alpha_*^2 = \alpha^2 \left(\frac{1}{2} + {\left(\frac{g_{\infty}}{m}\right)}^2 \right) .
\end{align*}
We show the lemma below:
\begin{lemma}[Key lemma I\hspace{-1.2pt}I]\label{lem:6} %
    Suppose that
    \begin{align*}
        0 \leq \varepsilon  < \frac{1}{\triangle T} .
    \end{align*}
    Then there exist $C_1(\varepsilon)$ and $C_2$ such that 
    \begin{align}\label{eq:lem6.1} %
        \mathbb{E}[e^{\varepsilon X}] \leq {C_1}(\varepsilon) e^{\varepsilon C_2 T} .
    \end{align}
    Here
    \begin{align*}
        {C_1}(\varepsilon)
        &= {\left(
            \frac{
                {\left(1-\varepsilon^2 T^2 \triangle^2\right)}^{- 1 / 2} - 1
            }{
                \varepsilon^2 T^2 \triangle^2 / 2
            }
        \right)}^{1 / 6}
        \exp(
            \frac{1}{4} + \varepsilon \alpha_*^2 \norm{ \frac{\hat{\varphi}}{\omega} }^2
        ) , \\
        C_2
        &= \alpha_*^2 \left(
            4 \norm{ \frac{\hat{\varphi}}{\omega} }^2
            + 5\norm{ \frac{\hat{\varphi}}{\sqrt{\omega}} }^2
            + \frac{1}{2} \norm{ \frac{\hat{\varphi}}{\omega} }^2
        \right)
        + 72 \alpha_*^4 \norm{ \frac{\hat{\varphi}}{\sqrt{\omega}} }^4.
    \end{align*}
    In particular
    \begin{align*}
        {\mathbb{E}[e^{\varepsilon X}]}^{1 / \varepsilon} \leq {C_3}(\varepsilon) e^{C_2 T} .
    \end{align*}
    Here
    \begin{align}\label{eq:lem6.2} %
        {C_3}(\varepsilon)
        = {\left(
            \frac{
                {\left( 1 - \varepsilon^2 T^2 \triangle^2 \right)}^{- \frac{1}{2}} - 1
            }{
                \varepsilon^2 T^2 \triangle^2 / 2
            }
        \right)}^{1 / (6 \varepsilon)}
        e^{1 / (4 \varepsilon)}
        e^{\alpha_*^2 \norm{ \hat{\varphi} / \omega }^2}.
    \end{align}
\end{lemma}
We show \cref{lem:7,lem:8,lem:9,lem:10} to prove \cref{lem:6}.
We see that
\begin{align*}
    X \leq \frac{\alpha^2}{2} \norm{ K_T }^2 + {\left(\frac{\abs{\alpha} g_{\infty}}{m} \right)}^2 \norm{ K_T }^2+ \frac{1}{4} .
\end{align*}
Then
\begin{align*}
    \mathbb{E}[e^{\varepsilon X}] \leq e^{1 / 4} \mathbb{E}[e^{\varepsilon \alpha_*^2 \norm{ K_T }^2}] .
\end{align*}
We set $W = (W_1, W_2, W_3)$ with
\begin{align*}
    W_{\mu} = {W_{\mu}}(T) = \int_0^T \mathrm{j}_s \tilde{\varphi}(\cdot - B_s) \, \mathrm{d}B_s^\mu . 
\end{align*}
Note that $\tilde{\varphi}$ is real-valued and that
$\mathrm{j}_s$ maps real-valued functions to real-valued functions.
Hence, $W_\mu$ is real-valued.
Let
\begin{align*}
    {\varrho_s}(k) = \mathrm{j}_s \tilde{\varphi}(\cdot - B_s)(k) .
\end{align*} 
\begin{lemma}\label{lem:7} %
    It follows that 
    \begin{align*}
        \mathbb{E}[{\exp}(\varepsilon \alpha_*^2 \norm{ K_T }^2)]
        \leq {\mathbb{E}[e^{3\varepsilon \alpha_*^2 I_{\mathrm{sing}}}]}^{1 / 3}
        {\mathbb{E}[e^{3\varepsilon \alpha_*^2 I_1}]}^{1 / 3}
        {\mathbb{E}[e^{3\varepsilon \alpha_*^2 I_2}]}^{1 / 3} ,
    \end{align*}
    where
    \begin{align*}
        I_{\mathrm{sing}}
        &= 2 \int_0^T {(\varrho_s, W(s))}_{L^2} \cdot \, \mathrm{d}B_s, \\
        I_1
        &= 2 \int_0^T \norm{ \varrho_s }_{L^2}^2 \, \mathrm{d}s, \\
        I_2
        &= - 2 \int_{\mathbb{R}^3} \frac{1}{\abs{k}^2}
            \int_0^T \varrho_s(k)(k \cdot W(s)) (k \cdot \, \mathrm{d}B_s)
        \, \mathrm{d}k.
    \end{align*}
\end{lemma}
\begin{proof}
    Since
    $\sum_{i = 1}^3 e_{\mu}^i e_{\nu}^i = d_{\mu \nu}(k) = \delta_{\mu \nu} - \frac{k_\mu k_\nu}{\abs{k}^2}$, we have
    \begin{align*}
        \norm{ {K}_T }_{\oplus^3 L^2}^2
        = \sum_{\mu, \nu = 1}^3 ({W_{\mu}}(T),d_{\mu \nu} {W_{\nu}}(T))_{L^2}
        = \sum_{\mu, \nu = 1}^3
            \int_{\mathbb{R}^3} {W_{\mu}}(T) {W_{\nu}}(T) {d_{\mu \nu}}(k) \, \mathrm{d}k.
    \end{align*}
    Employing the Ito formula we have 
    \begin{align*}
        \mathrm{d}(W_{\mu} W_{\nu})
        = \mathrm{d}{W_{\mu}} \cdot W_{\nu} + W_{\mu} \cdot \mathrm{d}{W_{\nu}} + 
        \mathrm{d}{W_{\mu}}\cdot \mathrm{d}{W_{\nu}},
    \end{align*}
    and hence
    \begin{align*}
        {W_{\mu}}(T) {W_{\nu}}(T)
        = \int_0^T \varrho_s {W_{\nu}}(s) \, \mathrm{d}B_s^\mu
        + \int_0^T \varrho_s {W_{\mu}}(s) \, \mathrm{d}B_s^\nu
        + \delta_{\mu \nu} \int_0^T \norm{ \varrho_s }_{L^2}^2 \, \mathrm{d}s.
    \end{align*}
    Therefore,
    \begin{align*}
\norm{ {K}_T }_{\oplus^3 L^2}^2
        &= 2 \int_0^T {(\varrho_s, W(s))}_{L^2} \cdot \, \mathrm{d}B_s
        + 2 \int_0^T \norm{ \varrho_s }_{L^2}^2 \, \mathrm{d}s \\
        &- 2 \int_{\mathbb{R}^3} \frac{1}{\abs{k}^2}
            \int_0^T {\varrho_s}(k) \cdot (k \cdot W(s)) (k \cdot \, \mathrm{d}B_s)
        \, \mathrm{d}k .
    \end{align*}
    Hence we have 
    \begin{align*}
        \mathbb{E}[\exp(\varepsilon \alpha_*^2 \norm{K_T}^2)]
        &= \mathbb{E}\left[
            \exp(
                \varepsilon \alpha_*^2
                \sum_{\mu, \nu = 1}^3
                    \int_{\mathbb{R}^3}
                        {W_{\mu}}(T) {W_{\nu}}(T)
                        {d_{\mu \nu}}(k)
                    \, \mathrm{d}k
            )
        \right] \\
        &\leq {\mathbb{E}[e^{3 \varepsilon \alpha_*^2 I_{\mathrm{sing}}}]}^{1/3}
        {\mathbb{E}[e^{3 \varepsilon \alpha_*^2 I_1}]}^{1/3}
        {\mathbb{E}[e^{3 \varepsilon \alpha_*^2 I_2}]}^{1/3} .
    \end{align*}
\end{proof}

The next series of lemmas is devoted to estimating 
$\mathbb{E}[e^{a I_{\mathrm{sing}}}]$, $ \mathbb{E}[e^{a I_1}]$ and $\mathbb{E}[e^{a I_2}]$, where
\begin{align*}
    a = 3 \varepsilon \alpha_*^2 .
\end{align*}
Estimating $\mathbb{E}[e^{a I_{\mathrm{sing}}}]$ is more difficult than
estimating $\mathbb{E}[e^{a I_1}]$ and $\mathbb{E}[e^{a I_2}]$.
The latter two can be represented in terms of Riemann integrals,
whereas the former cannot. 

\begin{lemma}\label{lem:8} %
    We have
    \begin{align}\label{eq:lem8} %
        \mathbb{E}[e^{a I_{\mathrm{sing}}}]
        \leq e^{6 {(2 a)}^2 T \norm{ {\hat{\varphi}} / {\sqrt{\omega}} }^4}
        {\left(
            \frac{
                {(1 - 16 {(2a)}^2 T^2 \norm{ \varphi /{\sqrt{\omega}} }^4)}^{- 1 / 2} -1
            }{
                8 {(2 a)}^2 T^2 \norm{ \varphi / {\sqrt{\omega}} }^4
            }
        \right)}^{1 / 2} .
    \end{align}
\end{lemma}
\begin{proof}
    Let
    \begin{align*}
        \beta = 2 a = 6 \varepsilon \alpha_*^2 .   
    \end{align*}
    By the Jensen inequality and exponential martingale property we have 
    \begin{align*}
        \mathbb{E}\left[
            e^{\beta \int_0^T (\varrho_s, W(s)) \cdot \, \mathrm{d}B_s}
        \right]
        \leq {\mathbb{E}\left[
            e^{2 \beta^2 \int_0^T {(\varrho_s, W(s))}^2 \, \mathrm{d}s}
        \right]}^{1 / 2}
        \leq {\left(
            \frac{1}{T}
            \int_0^T
                \mathbb{E}\left[
                    e^{2 \beta^2 T {(\varrho_s, W(s))}^2}
                \right]
            \, \mathrm{d}s
        \right)}^{1 / 2} .
    \end{align*}
    We estimate $\mathbb{E}\left[e^{2 \beta^2 T {(\varrho_s, W(s))}^2}\right]$.
    We have
    \begin{align*}
        {(\varrho_s, W(s))}_{L^2}^2
        = \sum_{\mu = 1}^3 {(\varrho_s,W_\mu(s))}_{L^2}^2
        = \sum_{\mu = 1}^3 {\left(
            \int_0^s
                {(\varrho_s, \varrho_t)}_{L^2} 
            \, \mathrm{d}B_t^\mu
        \right)}^2 .
    \end{align*}
    Let
    \begin{align}\label{eq:Qformal} %
        Q_{\mu} = {Q_{\mu}}(s)
        = \int_0^s
            {(\varrho_s, \varrho_t)}_{L^2} 
        \, \mathrm{d}B_t^\mu
    \end{align}
    and 
    \begin{align*}
        b = 2 \beta^2 T = 2 {(6 \varepsilon \alpha_*^2)}^2 T .
    \end{align*}
    We note that \cref{eq:Qformal} is only a formal expression, since the integrand $t \mapsto (\varrho_s, \varrho_t)$ is not adapted.
    More precisely, $Q_\mu$ is represented as
    \begin{align*}
        Q_\mu = 
        \int_{\mathbb{R}^3}
            \frac{\abs{\hat{\varphi}(k)}^2}{\omega(k)}
            e^{- i k \cdot B_s}
            \int_0^s
                e^{- (s - t) \omega(k)}
                e^{i k \cdot B_t} 
            \, \mathrm{d}B_t^\mu
        \, \mathrm{d}k .
    \end{align*}
    Then we see that
    \begin{align*}
        {\exp}(b \abs{Q}^2)
        = \frac{1}{{(4 \pi b)}^{3 / 2}}
        \int_{\mathbb{R}^3}
            \exp(x \cdot Q - \frac{\abs{x}^2}{4 b})
        \, \mathrm{d}x .
    \end{align*}
    Consequently,
    \begin{align}\label{eq:QGauss} %
        \mathbb{E}\left[
            \exp(b |Q|^2)
        \right]
        = \frac{1}{{(4 \pi b)}^{3 / 2}}
        \int_{\mathbb{R}^3}
            \exp(- \frac{\abs{x}^2}{4 b})
            \mathbb{E}\left[
                \exp(\sum_{\mu = 1}^3 x_\mu Q_\mu)
            \right]
        \, \mathrm{d}x .
    \end{align}
    We shall estimate
$\mathbb{E}\left[\exp(x_\mu Q_\mu)\right]$ 
    for a fixed $\mu$. However, we cannot directly apply the exponential
    martingale inequality, since $Q_\mu$ is not an Ito integral. 
    We therefore decompose $Q_\mu$ into a martingale part and a remainder term.

    For fixed $s > 0$ and $k \in \mathbb{R}^3$, set
    \begin{align*}
        {F_t}(k)
        = e^{- (s - t) \omega(k)} e^{- i k \cdot (B_s - B_t)}
        \qquad 0 \leq t \leq s.
    \end{align*}
    We define the backward stochastic integral by
    \begin{align*}
        \int_0^s {F_t}(k) \, \mathrm{d}\tilde{B}_t^\mu
        = \lim_{n \to \infty}
        \sum_{j = 0}^{n - 1}
            {F_{t_{j + 1}}}(k)
            (B_{t_{j + 1}}^\mu - B_{t_j}^\mu) ,
    \end{align*}
    where $0 = t_0 < t_1 < \cdots < t_n = s$ and $t_j = s j / n$. 
    Let
    \begin{align*}
        {G_t}(k) = e^{- (s - t) \omega(k)} e^{i k \cdot B_t} .
    \end{align*}
    Then
    \begin{align*}
        {F_t}(k) = e^{- i k \cdot B_s} {G_t}(k) .
    \end{align*}
    By Ito formula,
    \begin{align*}
        \mathrm{d}{G_t}(k)
        = {G_t}(k)\left\{
            i k \cdot \mathrm{d}B_t
            + \left(\omega(k) - \frac{\abs{k}^2}{2} \right) \mathrm{d}t
        \right\} .
    \end{align*}
    Hence we have the relationship between the stochastic integral 
    and the backward stochastic integral: 
    \begin{align*}
        \int_0^s {F_t}(k) \, \mathrm{d}\tilde{B}_t^\mu
        = e^{- i k \cdot B_s}
        \left\{
            \int_0^s {G_t}(k) \, \mathrm{d}B_t^\mu
            + i k_\mu \int_0^s {G_t}(k) \, \mathrm{d}t
        \right\} .
    \end{align*}
    Therefore,
    \begin{align*}
        &e^{- i k \cdot B_s}
        \int_0^s
            e^{- (s - t) \omega(k)}
            e^{i k \cdot B_t}
        \, \mathrm{d}B_t^\mu \\
        &= \int_0^s
            e^{- (s - t) \omega(k)}
            e^{- i k \cdot (B_s - B_t)}
        \, \mathrm{d}\tilde{B}_t^\mu
        - i k_\mu
        \int_0^s
            e^{- (s - t) \omega(k)}
            e^{- i k \cdot (B_s - B_t)}
        \, \mathrm{d}t .
    \end{align*}
    Consequently, setting
    \begin{align*}
        &\tilde{Q}_\mu
        = \int_{\mathbb{R}^3}
            \frac{\abs{\hat{\varphi}(k)}^2}{\omega(k)}
            \left(
                \int_0^s
                    e^{- (s - t) \omega(k)}
                    e^{- i k \cdot (B_s - B_t)}
                \, \mathrm{d}\tilde{B}_t^\mu
            \right)
        \, \mathrm{d}k, \\
        &R_\mu
        = - i
        \int_{\mathbb{R}^3}
            \frac{\abs{\hat{\varphi}(k)}^2}{\omega(k)}
            k_\mu
            \left(
                \int_0^s
                    e^{- (s - t) \omega(k)}
                    e^{- i k \cdot (B_s - B_t)}
                \, \mathrm{d}t
            \right)
        \mathrm{d}k ,
    \end{align*}
    we see that
    \begin{align*}
        Q_\mu = \tilde{Q}_\mu + R_\mu .
    \end{align*}
    Let $R = (R_1, R_2, R_3)$. 
    It is immediate to see that 
    \begin{align*}
        \abs{R}^2 \leq 3 \norm{ \frac{\hat{\varphi}}{\sqrt{\omega}} }^4 .
    \end{align*}
    We investigate $\tilde{Q}_\mu$.
    Let $0 = t_0 < t_1 < \cdots < t_n = s$ be a partition of $[0,s]$, and set $u_j = s - t_{n - j}$ for $j = 0, \ldots, n$.
    Then $0 = u_0 < u_1 < \cdots < u_n = s$.
    We also set $\hat{B}_u = B_s - B_{s - u}$ for $0 \leq u \leq s$.
    Then ${(\hat{B}_u)}_{0 \leq u \leq s}$ is a Brownian motion and 
    \begin{align*}
        \int_0^s
            e^{- (s - t) \omega(k)}
            e^{- i k \cdot (B_s - B_t)}
        \, \mathrm{d}\tilde{B}_t^\mu
        &= \lim_{n \to \infty}
        \sum_{i = 0}^{n - 1}
            e^{- (s - t_{i + 1}) \omega(k)}
            e^{- i k \cdot (B_s - B_{t_{i + 1}})}
            (B_{t_{i + 1}}^\mu - B_{t_i}^\mu) \\
        &= \lim_{n \to \infty}
        \sum_{j = 0}^{n-1}
            e^{- u_j \omega(k)}
            e^{- i k \cdot(B_s - B_{s - u_j})}
            (B_{s - u_j}^\mu - B_{s - u_{j + 1}}^\mu) .
    \end{align*}
    Since $\hat{B}_{u_{j + 1}}^\mu - \hat{B}_{u_j}^\mu = B_{s - u_j}^\mu - B_{s - u_{j + 1}}^\mu$, we obtain
    \begin{align*}
        \int_0^s
            e^{- (s - t) \omega(k)}
            e^{- i k \cdot(B_s - B_t)}
        \, \mathrm{d}\tilde{B}_t^\mu
        &= \lim_{n \to \infty}
        \sum_{j = 0}^{n - 1}
            e^{- u_j \omega(k)}
            e^{- i k \cdot \hat{B}_{u_j}}
            (\hat{B}_{u_{j + 1}}^\mu - \hat{B}_{u_j}^\mu) \\
        &= \int_0^s
            e^{- u \omega(k)}
            e^{-i k \cdot \hat{B}_u}
        \, \mathrm{d}\hat{B}_u^\mu.
    \end{align*}
    Let
    \begin{align*}
        \ev*{\tilde{Q}_\mu}
        = \int_0^s \left(
            \int_{\mathbb{R}^3 \times \mathbb{R}^3}
                \frac{\abs{\hat{\varphi}(k)}^2}{\omega(k)}
                \frac{\abs{\hat{\varphi}(k')}^2}{\omega(k')}
                e^{- u (\omega(k) + \omega(k'))}
                e^{- i (k + k') \cdot \hat{B}_u}
            \, \mathrm{d}k \mathrm{d}k'
        \right) \, \mathrm{d}u .
    \end{align*}
    We see the bound:
    \begin{align}\label{eq:Qbound1} %
        \ev*{\tilde{Q}_\mu}
        \leq \int_{\mathbb{R}^3 \times \mathbb{R}^3}
            \frac{\abs{\hat{\varphi}(k)}^2}{\omega(k)}
            \frac{\abs{\hat{\varphi}(k')}^2}{\omega(k')}
            \frac{1}{\omega(k) + \omega(k')}
            (1 - e^{- s (\omega(k) + \omega(k'))})
        \, \mathrm{d}k \, \mathrm{d}k'
        \leq s \norm{ \frac{\hat{\varphi}}{\sqrt{\omega}} }^4.
    \end{align}
    By the exponential martingale property of $\tilde{Q}_\mu$ we have
    \begin{align}\label{eq:Qexp} %
        \mathbb{E}[e^{x_\mu \tilde{Q}_\mu}]
        \leq {\mathbb{E}[e^{2 x_\mu^2 \ev*{\tilde{Q}_\mu}}]}^{1/2}
        \leq e^{x_\mu^2 s \norm{ {\hat{\varphi}} / {\sqrt{\omega}} }^4}.
    \end{align}
    By \cref{eq:QGauss} and \cref{eq:Qexp} we can obtain that 
    \begin{align*}
        \mathbb{E}[e^{b \abs{Q}^2}]
        &\leq \mathbb{E}[e^{2b \abs*{\tilde{Q}}^2 + 2 b \abs{R}^2}] 
        \leq \mathbb{E}[e^{2b \abs*{\tilde{Q}}^2}] e^{6 b \norm{ {\hat{\varphi}} / {\sqrt{\omega}} }^4} \\
        &\leq \frac{e^{6 b \norm{ {\hat{\varphi}} / {\sqrt{\omega}} }^4}}{{(8 \pi b)}^{3 / 2}}
        \int_{\mathbb{R}^3}
            e^{- {\abs{x}^2} / (8 b)}
            \prod_{\mu = 1}^3 \mathbb{E}[e^{x_\mu Q_\mu}]
        \, \mathrm{d}x \\
        &\leq \frac{e^{6 b \norm{ {\hat{\varphi}} / {\sqrt{\omega}} }^4}}{{(8 \pi b)}^{3 / 2}}
        \int_{\mathbb{R}^3}
            e^{- {\abs{x}^2} / (8 b)}
            e^{s \norm{ {\hat{\varphi}} / {\sqrt{\omega}} }^4 \abs{x}^2}
            \, \mathrm{d}x\\
        &= e^{6 b\norm{ {\hat{\varphi}} / {\sqrt{\omega}} }^4}
        {\left(1 - 8 b s \norm{ \frac{\hat{\varphi}}{\sqrt{\omega}} }^4\right)}^{- 3 / 2}.
    \end{align*}
    Then
    \begin{align*}
        \frac{1}{T} \int_0^T
            {\left(1 - 8 b s \norm{\frac{\hat{\varphi}}{\sqrt{\omega}}}^4 \right)}^{- 3 / 2}
        \, \mathrm{d}s
        = \frac{1}{4 b T \norm{{\hat{\varphi}} / {\sqrt{\omega}}}^4}
        \left\{
            {\left(1 - 8 b T \norm{\frac{\hat{\varphi}}{\sqrt{\omega}}}^4\right)}^{- 1 / 2}
            - 1
        \right\} .
    \end{align*}
    Finally we have 
    \begin{align*}
        \mathbb{E}\left[
            e^{\beta \int_0^T (\varrho_s, W) \cdot \,\mathrm{d} B_s}
        \right]
        \leq
        e^{6 \beta^2 T \norm{\hat{\varphi} / \sqrt{\omega}}^4}{\left(
            \frac{
                {\left(1 - 16 \beta^2 T^2 \norm{\hat{\varphi} / \sqrt{\omega}}^4 \right)}^{- 1 / 2} - 1
            }{
                8 \beta^2 T^2 \norm{\hat{\varphi} / \sqrt{\omega}}^4
            }
        \right)}^{1 / 2} .
    \end{align*}
    Then we have \cref{eq:lem8}.
\end{proof}

Next we estimate $\mathbb{E}[e^{a I_1}]$. 
\begin{lemma}\label{lem:9} %
    We have 
    \begin{align}\label{eq:lem9} %
        \mathbb{E}[e^{a I_1}]
        \leq \exp(2 T a \norm{\frac{\hat{\varphi}}{\sqrt{\omega}}}^2) .
    \end{align}
\end{lemma}
\begin{proof}
    We can see that $I_1 = 2 \int_0^T \norm{\varrho_s}^2 \, \mathrm{d}s = 2 T \norm{\hat{\varphi} / \sqrt{\omega}}^2$.
    Then \cref{eq:lem9} follows. 
\end{proof}

Finally we estimate $\mathbb{E}[e^{a I_2}]$. 
\begin{lemma}\label{lem:10} %
    We have 
    \begin{align}\label{eq:lem10} %
        \abs{\mathbb{E}[e^{a I_2}]}
        \leq \exp(
            a \norm{\frac{\hat{\varphi}}{\omega}}^2
            + 2 a T \left(
                2 \norm{\frac{\hat{\varphi}}{\omega}}^2
                + \frac{3}{2} \norm{\frac{\hat{\varphi}}{\sqrt{\omega}}}^2
                + \frac{1}{4} \norm{\hat{\varphi}}^2
            \right)
        ) .
    \end{align}
\end{lemma}
\begin{proof}
    Let
    \begin{align*}
        &Y_s(k)
        = {\varrho_s}(k) k \cdot W(s)
        = \int_0^s
            \frac{\abs{\hat{\varphi}(k)}^2}{\omega(k)}
            e^{-(s-u)\omega(k)}
            e^{-ik\cdot (B_s-B_u)}
        (k \cdot \, \mathrm{d}B_u) , \\
        &Z_u
        = e^{- (s - u) \omega(k)}
        e^{- i k \cdot (B_s - B_u)} ,
        \qquad 0 \leq u \leq s .
    \end{align*}
    Then
    \begin{align*}
        I_2
        = - \int_{\mathbb{R}^3}
            \frac{2}{\abs{k}^2}
            \int_0^T
                Y_s(k)
            (k \cdot \, \mathrm{d}B_s)
        \, \mathrm{d}k .
    \end{align*}
    Since
$
        e^{- i k \cdot (B_s - B_u)}
        = e^{- i k \cdot B_s} e^{i k \cdot B_u}$, 
    we apply Ito formula to
    \begin{align*}
        F(u, x)
        = e^{- (s - u) \omega(k)} e^{- i k \cdot (B_s - x)} .
    \end{align*}
    Since
$        \partial_u F
        = \omega(k) F$,
$\nabla_x F = i k F$ and 
$\Delta_x F = -\abs{k}^2 F$,
    Ito formula yields
    \begin{align*}
        dZ_u
        &= \partial_u F(u, B_t) \, \mathrm{d}t
        + \nabla_x F(u, B_t) \cdot \, \mathrm{d}B_u
        + \frac{1}{2} \Delta_x F(u, B_u) \mathrm{d}u \\
        &= \left(\omega(k) - \frac{\abs{k}^2}{2}\right) Z_u \mathrm{d}u
        + i Z_u k \cdot \, \mathrm{d}B_u .
    \end{align*}
    Hence
    \begin{align*}
        Z_u k \cdot \, \mathrm{d}B_u
        = \frac{1}{i} \mathrm{d}Z_u
        - \frac{1}{i} \left(\omega(k) - \frac{\abs{k}^2}{2}\right) Z_u \mathrm{d}u .
    \end{align*}
    Integrating over $[0,s]$, we obtain
    \begin{align*}
        \int_0^s
            e^{- (s - u) \omega(k)}
            e^{- i k \cdot (B_s - B_u)}
        k \cdot \, \mathrm{d}B_u 
        &= \frac{1}{i} \left(Z_s - Z_0\right)
        - \frac{1}{i}
        \int_0^s
            \left(\omega(k) - \frac{\abs{k}^2}{2} \right) Z_u
        \mathrm{d}u \\
       & = i (R_s - 1),
    \end{align*}
    where we set
    \begin{align*}
        R_s = e^{- s \omega(k)} e^{- i k \cdot B_s}
        + \int_0^s
            \left(\omega(k) - \frac{\abs{k}^2}{2}\right)
            e^{- (s - u) \omega(k)}
            e^{- i k \cdot (B_s - B_u)}
        \, \mathrm{d}u .
    \end{align*}
    Thus we obtain the identity:
    \begin{align*}
        I_2
        = - 2 \int_{\mathbb{R}^3}
            \frac{\abs{\hat{\varphi}(k)}^2}{\abs{k}^2\omega(k)}
            \int_0^T i (R_s - 1) (k \cdot \, \mathrm{d}B_s)
        \,\mathrm{d}k .
    \end{align*}
    We shall compute $\int_0^T i (R_s - 1) (k \cdot \, \mathrm{d}B_s)$. 
    Since
    \begin{align*}
        \mathrm{d}R_s
        = - i R_s(k \cdot \, \mathrm{d}B_s)
        - \left(\omega(k) + \frac{\abs{k}^2}{2} \right) R_s \, \mathrm{d}s
        + \left(\omega(k)-\frac{\abs{k}^2}{2}\right) \, \mathrm{d}s ,
    \end{align*}
    we have
    \begin{align*}
        i R_s (k \cdot \, \mathrm{d}B_s)
        = - \mathrm{d}R_s
        - \left(\omega(k) + \frac{\abs{k}^2}{2} \right) R_s \, \mathrm{d}s
        + \left(\omega(k) - \frac{\abs{k}^2}{2} \right) \, \mathrm{d}s .
    \end{align*}
    Hence
    \begin{align*}
        &\int_0^T i (R_s - 1) (k \cdot \, \mathrm{d}B_s) \\
        &= - \left(R_T - 1 \right) 
        - \int_0^T
            \left(\omega(k) + \frac{\abs{k}^2}{2} \right) R_s
        \, \mathrm{d}s
        + T \left(\omega(k) - \frac{\abs{k}^2}{2} \right)
        - i k \cdot B_T .
    \end{align*}
    We obtain the decomposition: 
    \begin{align*}
        I_2
        = 2 U_T - 2 i
        \int_{\mathbb{R}^3}
            \frac{\abs{\hat{\varphi}(k)}^2}{\omega(k) \abs{k}^2} k \cdot B_T
        \, \mathrm{d}k ,
    \end{align*}
    where
    \begin{align*}
        U_T
        =& - \int_{\mathbb{R}^3}
            \frac{\abs{\hat{\varphi}(k)}^2}{\omega(k) \abs{k}^2}
            \left({R_t}(k) - 1 \right)
        \, \mathrm{d}k 
        - \int_{\mathbb{R}^3}
            \int_0^T
                \frac{\abs{\hat{\varphi}(k)}^2}{\omega(k) \abs{k}^2}
                \left(\omega(k) + \frac{\abs{k}^2}{2} \right)
                R_s(k)
            \, \mathrm{d}s
        \, \mathrm{d}k \\
        &+ T \int_{\mathbb{R}^3}
            \frac{\abs{\hat{\varphi}(k)}^2}{\omega(k) \abs{k}^2}
            \left(\omega(k) - \frac{\abs{k}^2}{2}\right)
        \, \mathrm{d}k . 
    \end{align*}
    We note that $U_T$ involves no stochastic integrals. 
    Since
    \begin{align*}
        \abs{R_s(k)}
        \leq 1 + \frac{\abs{k}^2}{2 \omega(k)} ,
    \end{align*}
    we have
    \begin{align*}
        \abs{U_T}
        &\leq \int_{\mathbb{R}^3}
            \frac{\abs{\hat{\varphi}(k)}^2}{\abs{k}^2}
            \frac{\abs{k}^2}{2 {\omega(k)}^2}
        \, \mathrm{d}k
        + \int_{\mathbb{R}^3}
            \int_0^T
                \frac{\abs{\hat{\varphi}(k)}^2}{\omega(k) \abs{k}^2}
                \left(\omega(k) + \frac{\abs{k}^2}{2} \right)
                \left(1 + \frac{\abs{k}^2}{2\omega(k)} \right) 
            \, \mathrm{d}t
        \, \mathrm{d}k \\
        &\quad + T \int_{\mathbb{R}^3}
            \frac{\abs{\hat{\varphi}(k)}^2}{\omega(k) \abs{k}^2}
            \abs{\left(\omega(k) - \frac{\abs{k}^2}{2} \right)}
        \, \mathrm{d}k \\
        &\leq \frac{1}{2} \norm{\frac{\hat{\varphi}}{\omega}}^2
        + T \left(2 \norm{\frac{\hat{\varphi}}{\omega}}^2
        + \frac{3}{2} \norm{\frac{\hat{\varphi}}{\sqrt{\omega}}}^2
        + \frac{1}{4} \norm{\hat{\varphi}}^2 \right)
    \end{align*}
    Then we have
    \begin{align*}
        \abs{\mathbb{E}[e^{a I_2}]}
        &= \abs{\mathbb{E}\left[
            \exp(2 a U_T
                - 2 i
                \int_{\mathbb{R}^3}
                    \frac{\abs{\hat{\varphi}(k)}^2}{\abs{k}^2} k \cdot B_t
                \, \mathrm{d}k
            )
        \right]}\\
        &\leq \mathbb{E}[\exp(2 a U_T)] 
        \leq \exp(
            a \norm{\frac{\hat{\varphi}}{\omega}}^2
            + 2 a T \left(2 \norm{\frac{\hat{\varphi}}{\omega}}^2
            + \frac{3}{2} \norm{\frac{\hat{\varphi}}{\sqrt{\omega}}}^2
            + \frac{1}{4} \norm{\hat{\varphi}}^2 \right)
        )
    \end{align*}
    and the lemma is proven. 
\end{proof}

{Proof of \cref{lem:6}}

\begin{proof}
    Combining \cref{eq:lem8}, \cref{eq:lem9} and \cref{eq:lem10}, we see that 
    \begin{align*}
        \mathbb{E}[e^{\varepsilon X}]
        \leq& {\left(
            \frac{
                {\left(
                    1 - 16 {(2 a)}^2 T^2 \norm{{\hat{\varphi}} / {\sqrt{\omega}}}^4
                \right)}^{- 1 / 2} - 1
            }{
                8 {(2 a)}^2 T^2 \norm{ {\hat{\varphi}} / {\sqrt{\omega}}}^4
            }
        \right)}^{1 / 6}
        \exp(
            \frac{1}{4} + \frac{a}{3} \norm{\frac{\hat{\varphi}}{\omega}}^2
        ) \\
        &\times \exp(
            2 {(2 a)}^2 T \norm{\frac{\hat{\varphi}}{\sqrt{\omega}}}^4
        )
        \exp(
            2 a T
            \left(
                \frac{2}{3} \norm{\frac{\hat{\varphi}}{\omega} }^2
                + \frac{5}{6} \norm{\frac{\hat{\varphi}}{\sqrt{\omega}}}^2
                + \frac{1}{12} \norm{\hat{\varphi}}^2
            \right)
        ) \\
        =& \mathrm{RHS of } \cref{eq:lem6.1},
    \end{align*}
    where $a = 3 \varepsilon \alpha_*^2$.
    Note that
    ${\left(1 - 16 {(2a)}^2 T^2 \norm{\hat{\varphi} / \sqrt{\omega}}^4 \right)}^{- 1 / 2}>1$ if and only if $\varepsilon < 1 / (\triangle T)$.
    Then the proof is complete.
\end{proof}

We can obtain an alternative bound from \cref{eq:Qbound1}.
\begin{lemma}
    Let
    \begin{align*}
        \tilde{\triangle}
        = 12 \sqrt{2} \alpha_*^2 \norm{\frac{\hat{\varphi}}{\omega^{3/4}}}^2 .
    \end{align*}
    Suppose that
    \begin{align*}
        0 \leq \varepsilon < \frac{1}{\tilde{\triangle} \sqrt{T}}.
    \end{align*}
    Then
    \begin{align*}
        \mathbb{E}[e^{\varepsilon X}]
        \leq \tilde {C_1}(\varepsilon) e^{C_2 \varepsilon T},
    \end{align*}
    where $C_2$ is the same constant as in \cref{lem:6}, and
    \begin{align*}
        {\tilde{C}_1}(\varepsilon)
        = {\left(
            1 
            - \varepsilon^2 T \tilde{\triangle}^2
        \right)}^{- 1 / 4}
        \exp(
            \frac{1}{4}
            + \varepsilon \alpha_*^2 \norm{\frac{\hat{\varphi}}{\omega}}^2
        ) .
    \end{align*}
\end{lemma}
\begin{proof}
    We first obtain the following bound:
    \begin{align*}
        \ev*{\tilde{Q}_\mu}
        \leq \int_{\mathbb{R}^3 \times \mathbb{R}^3}
            \frac{\abs{\hat{\varphi}(k)}^2}{\omega(k)}
            \frac{\abs{\hat{\varphi}(k')}^2}{\omega(k')}
            \frac{1 - e^{- s (\omega(k) + \omega(k'))}}{\omega(k) + \omega(k')}
        \, \mathrm{d}k \, \mathrm{d}k'
        \leq \frac{1}{2} \norm{\frac{\hat{\varphi}}{\omega^{3/4}}}^4 .
    \end{align*}
    In particular, this bound is uniform in $s$.
    Proceeding as in the proof of \cref{eq:lem8}, the exponential martingale property yields
    \begin{align*}
        \mathbb{E}\left[e^{2 b \abs{\tilde{Q}}^2}\right]
        \leq {\left(1 - 4 b \norm{\frac{\hat{\varphi}}{\omega^{3/4}}}^4
        \right)}^{- 3 / 2} .
    \end{align*}
    Since the right-hand side is independent of $s$, the assertion now follows by proceeding as in the proof of \cref{lem:6}.
\end{proof}

\subsection{Lower bounds of exponential decay}
The exponential decay of bound states is a fundamental problem in the spectral theory of Schr\"{o}dinger operators and quantum field models.
While upper bounds on the decay of bound states have been extensively studied, lower bounds are in general more delicate. In particular, the standard approach based on the Agmon distance does not seem to be directly applicable to the lower bounds considered here.
We therefore adopt a different approach to derive lower bounds on the exponential decay of bound states.

\begin{lemma}
    Suppose that $1 / p + 1 / q = 1$ and 
    \begin{align}\label{eq:lem12.1} %
        \frac{q}{p} < \frac{1}{\triangle T}. 
    \end{align}
    Then    
    \begin{align}\label{eq:lem12.2} %
        {(\mathbbm{1}, {\varphi_\mathrm{g}}(x))}_\mathcal{F} 
        \geq \frac{e^{{C_2} T}}{{C_3}(q / p)}
        {\mathbb{E}\left[
            \mathbbm{1}_{S_T}
            e^{- \frac{1}{p} \int_0^T (V(B_t + x) - E_\alpha) \, \mathrm{d}t}
            {m(B_T + x)}^{1 / p}
        \right]}^p ,
    \end{align}
    where ${C_3}(\cdot)$ is given by \cref{eq:lem6.2}:
    \begin{align}
        {C_3}\left(\frac{q}{p} \right)
        = {\left(
            \frac{
                {\left(1 - {(q / p)}^2 T^2 \triangle^2 \right)}^{- 1 / 2} - 1
            }{
                {(q / p)}^2 T^2 \triangle^2 / 2
            }
        \right)}^{1 / (6 q / p)}
        e^{1 / (4 q / p)}
        e^{\alpha_*^2 \norm{{\hat{\varphi}} / {\omega}}^2} .
    \end{align}
\end{lemma}
\begin{proof}
    Shifting the starting point of Brownian motion to the origin, we have
    \begin{align*}
        {(\mathbbm{1}, {\varphi_\mathrm{g}}(x))}_\mathcal{F} 
        \geq \mathbb{E}\left[
            \mathbbm{1}_{S_T}
            e^{- \int_0^T (V(B_t + x) - E_\alpha ) \, \mathrm{d}t}
            e^{- X}
            m(B_T + x)
        \right] .
    \end{align*}
    In general by the H\"{o}lder inequality we obtain that for $1 / p + 1 / q = 1$
    \begin{align*}
        \mathbb{E}[f^{1 / p}]
        = \mathbb{E}[{(f g)}^{1 / p} g^{- 1 / p}]
        \leq {\mathbb{E}[f g]}^{1 / p}
        {\mathbb{E}[g^{- q / p}]}^{1 / q}
    \end{align*}
    and then
    \begin{align}\label{eq:holder} %
        \mathbb{E}[f g]
        \geq \frac{
            {\mathbb{E}[f^{1 / p}]}^p
        }{
            {\mathbb{E}[g^{- q / p}]}^{p / q}
        } .
    \end{align}
    By \cref{eq:holder}, we have the lower bound:
    \begin{align*}
        {(\mathbbm{1}, {\varphi_\mathrm{g}}(x))}_\mathcal{F} 
        \geq \frac{
            {\mathbb{E}[
                \mathbbm{1}_{S_T}
                e^{- \frac{1}{p} \int_0^T (V(B_t+x) - E_\alpha ) \, \mathrm{d}t}
                {m(B_T + x)}^{1 / p}
            ]}^p
        }{
            {\mathbb{E}[
                e^{{\frac{q}{p}} X}
            ]}^{p / q}
        } .
    \end{align*}
    By \cref{eq:lem12.1} and \cref{lem:6},
    \begin{align*}
        {\mathbb{E}\left[e^{\frac{q}{p} X}\right]}^{p / q}
        \leq {C_3}\left(\frac{q}{p}\right) e^{- {C_2} T}, 
    \end{align*}
    Then \cref{eq:lem12.2} follows. 
\end{proof}

We derive an exponential lower bound of the ground state of $H_\alpha$.
First we recall some estimates of probabilities of sets of paths. 
Let $a > 0$, $c < d$, and 
\begin{align*}
    {P_j}(a, [c, d], T) = \left\{\sup_{0 \leq s \leq T} \abs{B_s^j} \leq a \right\} \cap \left\{B_T^j \in [c, d] \right\} ,
    \qquad j = 1, 2, 3 .
\end{align*}
Let
\begin{align*}
    \rho(\xi)
    = 1
    - e^{- 21 \xi / 8}
    - e^{- 5 \xi / 8}
    - e^{- 165 \xi / 8} .
\end{align*}
The following statement has been established. 
\begin{lemma}
    Suppose that $a > 0$, $\alpha > 0$ and $T > 0$ satisfy that $\alpha < a / 2$ and $a^2 / T > \beta$. 
    Here $\beta$ is the unique positive solution of $\rho(\xi) = 0$.  
    Then for every $x \in [- (a - \alpha), a - \alpha]$ we have
    \begin{align*}
        \mathcal{W}({P_j}(a, [x - \alpha, x + \alpha], T))
        \geq \frac{\alpha}{\sqrt{2 \pi T}}
        \rho\left( \frac{a^2}{t} \right)
        e^{- a^2 / 2} ,
        \qquad j = 1, 2, 3 .
    \end{align*}
\end{lemma}
\begin{proof}
    We refer the reader to see~\cite[\text{1.15.8, p.174}]{BS02},~\cite{car78} and~\cite[\text{Section 2.9.2}]{HL26}.
\end{proof}
Let 
    \begin{align*}
        {W_a}(x) = \sup\{W(y) \mid \abs{y_j - x_j} < a_j, j = 1, 2, 3\}. 
    \end{align*}
\begin{theorem}\label{thm:14} %
    Suppose that $a_j, \alpha_j, b_j > 0$ and $T$ satisfy that ${a_j} / 2 > \alpha_j$, $a_j^2 / T > \beta$, and
    \begin{align*}
        \abs{[- a_j, a_j] \cap [- x_j - b_j, - x_j + b_j]} > 2 \alpha_j
    \end{align*}
    for $j = 1, 2, 3$.
    Here $\abs{A}$ denotes the Lebesgue measure of the set $A$.
    Let $1 / p + 1 / q = 1$, $p, q > 1$ and
    \begin{align*}
        \frac{q}{p} < \frac{1}{\triangle T}.
    \end{align*}
    Suppose that $V = W - U \in \mathcal{V}_{\mathrm{lower}}$. 
    Then
    \begin{align}\label{eq:thm14} %
        {(\mathbbm{1}, {\varphi_\mathrm{g}}(x))}_\mathcal{F}
        \geq \frac{e^{C_2 T}}{C_3(q / p)}
        e^{T E_\alpha}
        e^{- \int_0^T W_a(x) \, \mathrm{d}t}
        \ell_K
        {\left(
            \prod_{j = 1}^3
                \frac{\alpha_j}{\sqrt{2 \pi T}}
                \rho\left(\frac{a_j^2}{T}\right)
                e^{- a_j^2 / (2 T)}
        \right)}^p .
    \end{align}
    where
$\ell_K$ is given by $\ell_K = \inf_{B_T + x \in K}(\mathbbm{1}, {\varphi_\mathrm{g}}(x + B_T))$ for $K = [- b_1, b_1] \times [- b_2, b_2] \times [- b_3, b_3]$.
\end{theorem}
\begin{proof}
    Let
    \begin{align*}
        P
        = \bigcap_{j = 1}^3 {P_j}(\alpha_j, [- b_j - x_j, b_j - x_j], T) .
    \end{align*}
    We see that for $\alpha_j < a_j / 2, a_j^2 / t > \beta$ and
    \begin{align*}
        \mathcal{W}\left(
            \bigcap_{j = 1}^3
                {P_j}(\alpha_j, [k_j, k_j + 2 \alpha_j], T)
        \right)
        \geq \prod_{j = 1}^3
            \frac{\alpha_j}{\sqrt{2 \pi t}}
            \rho\left(\frac{a_j^2}{T} \right)
            e^{- a_j^2 / (2 T)}
    \end{align*}
    for any $k_j \in [- a_j, a_j]$.
    By the assumption there exists $k_j \in [- a_j, a_j]$ such that
    \begin{align*}
        [k_j, k_j + 2 \alpha_j] \subset [- b_j - x_j, b_j - x_j] .
    \end{align*}
    We have
    \begin{align*}
        \mathcal{W}\left(P
        \right)
        \geq \mathcal{W}\left(
            \bigcap_{j = 1}^3
                {P_j}(\alpha_j, [k_j, k_j + 2 \alpha_j], T)
        \right) 
        \geq \prod_{j = 1}^3
            \frac{\alpha_j}{\sqrt{2 \pi T}}
            \rho\left(\frac{a_j^2}{T} \right)
            e^{- a_j^2 / (2 T)} .
    \end{align*}
    Since
    \begin{align*}
        {(\mathbbm{1}, {\varphi_\mathrm{g}}(x))}_\mathcal{F} 
        \geq \frac{e^{C_2 T}}{{C_3}(q/p)}   
        {\mathbb{E}[
            \mathbbm{1}_P
            e^{- \frac{1}{p} \int_0^T (V(B_t+x) - E_\alpha) \, \mathrm{d}t}
            {m(B_T + x)}^{1 / p}
        ]}^p ,
    \end{align*}
    we see that
    \begin{align*}
        {(\mathbbm{1}, {\varphi_\mathrm{g}}(x))}_\mathcal{F} 
        &\geq \frac{e^{C_2 T}}{{C_3}(q/p)}
        e^{T E_\alpha}
        {\mathbb{E}[
            \mathbbm{1}_P
            e^{- \frac{1}{p} \int_0^T W_a(x) \, \mathrm{d}t}
            \ell_K^{1 / p}
        ]}^p \\
        &\geq \frac{e^{C_2 T}}{{C_3}(q / p)}
        e^{T E_\alpha}
        e^{- \int_0^T W_a(x) \, \mathrm{d}t}
        \ell_K
        {\mathbb{E}[\mathbbm{1}_P]}^p \\
        &\geq \frac{e^{C_2 T}}{{C_3}(q / p)}
        e^{T E_\alpha}
        e^{- \int_0^T W_a(x) \, \mathrm{d}t}
        \ell_K  
        {\left(
            \prod_{j = 1}^3
                \frac{\alpha_j}{\sqrt{2 \pi T}}
                \rho\left(\frac{a_j^2}{T}\right)
                e^{- a_j^2 / (2 T)}
        \right)}^p.
    \end{align*}
    Then the theorem follows. 
\end{proof}

\begin{corollary}\label{maincoro}
Suppose that, for some real number $n>0$ and constants
$0<A\leq B$ and $c_1,c_2>0$,
\[
    A^2|x|^{2n}-c_1
    \leq V(x)
    \leq B^2|x|^{2n}+c_2,
    \qquad x\in\mathbb{R}^3.
\]
Let $0<\epsilon<b<1/2$.
Then there exist constants $D_\epsilon,C_\epsilon>0$
such that
\[
    D_\epsilon e^{-(\sqrt2 2^n B+\epsilon)|x|^{n+1}}
    \leq \|\varphi_{\mathrm g}(x)\|_{\mathcal F}
    \leq
    C_\epsilon e^{-(b-\epsilon)\frac{A}{2^{n+4}}|x|^{n+1}},
    \qquad x\in\mathbb{R}^3.
\]
\end{corollary}

\begin{proof}
The upper bound follows from 
Corollary~\ref{cupper}. 
Use the decomposition $V=W-U$ with $W=V_+$ and $U=V_-$.
For $r=|x|\geq1$, set
\[
    T=\frac{1}{\sqrt2 2^n B}r^{1-n},
    \qquad
    a_j=1+|x_j|+\sqrt r,
    \qquad
    \alpha_j=\frac12,
    \qquad
    b_j=1,
    \qquad j=1,2,3.
\]
If $|y_j-x_j|<a_j$ for $j=1,2,3$, then
$|y|\leq2r+\sqrt3(1+\sqrt r)$.
Hence
\begin{align*}
    W_a(x)
    &\leq B^2\left(
        \sum_{j=1}^3(2|x_j|+1+\sqrt r)^2
    \right)^n+c_2\\
    &\leq B^2\bigl(2r+\sqrt3(1+\sqrt r)\bigr)^{2n}+c_2\\
    &=B^2 2^{2n}r^{2n}+o(r^{2n}). 
\end{align*}
It follows that
\begin{align}
\label{tu1}
    TW_a(x)
    \leq \frac{ 2^nB}{\sqrt2 }r^{n+1}+o(r^{n+1}).
\end{align}
Moreover,
$
    \frac{a_j^2}{T}\geq\sqrt2 2^n B r^n\to \infty$ as $r\to\infty$ for 
    $j=1,2,3$. 
Since
\[
    \sum_{j=1}^3a_j^2
    =
    r^2+2(1+\sqrt r)\sum_{j=1}^3|x_j|
    +3(1+\sqrt r)^2
    =r^2+o(r^{2}),
\]
we have
\begin{align}
\label{tu2}
    \frac1{2T}\sum_{j=1}^3a_j^2
    =\frac{2^nB}{\sqrt2}r^{n+1}+o(r^{n+1}).
\end{align}
Applying \cref{eq:thm14}, we obtain
\begin{align*}
    (\mathbbm{1},\varphi_{\mathrm g}(x))_{\mathcal F}
    &\geq
    \frac{\ell_K}{C_3(q/p)}
    \left(\frac{T^{-1/2}}{2\sqrt{2\pi}}\right)^3
    \prod_{j=1}^3\rho(a_j^2/T)\\ 
&\times     \exp((E_\alpha+C_2)T)
\exp\left(
        -TW_a(x)
        -\frac1{2T}\sum_{j=1}^3a_j^2
    \right).
\end{align*}
Note that $\lim_{r\to\infty}C_3(q/p)=1$.
The estimates for $C_3(q/p)$ and $\rho$ give
\[
    \frac1{C_3(q/p)}\geq\frac12,
    \qquad
    \rho(a_j^2/T)\geq\frac12,
    \qquad j=1,2,3,
\]
for all sufficiently large $r$.
Since $T=o(r^{n+1})$, by \eqref{tu1} and \eqref{tu2}  we consequently obtain,
for some constant $c>0$,
\[
    (\mathbbm{1},\varphi_{\mathrm g}(x))_{\mathcal F}
    \geq
    c\,r^{3(n-1)/2}
    \exp\left\{
        -\sqrt2 2^n B r^{n+1}+o(r^{n+1})
    \right\}.
\]
Then 
there exist constants
$R_\epsilon,c_*>0$ such that
\[
    (\mathbbm{1},\varphi_{\mathrm g}(x))_{\mathcal F}
    \geq c_*e^{-(\sqrt2\,B2^n+\epsilon)|x|^{n+1}},
    \qquad |x|\geq R_\epsilon.
\]
Finally,
$(\mathbbm{1},\varphi_{\mathrm g}(x))_{\mathcal F}$
is strictly positive and continuous in $x$.
Thus the lower bound extends to all $x\in\mathbb{R}^3$
after decreasing the constant.
This completes the proof.
\end{proof}

\section{The dipole approximation and the Agmon metric}\label{sec:4} %
\subsection{Estimates of exponential moments of double stochastic integrals}
In this section, we consider the Pauli-Fierz Hamiltonian under the dipole approximation. It is obtained by replacing $A(x)$ with $A(0)$.
Then
\begin{align*}
    H_\alpha(0) = \frac{1}{2}{(-i\nabla - \alpha A(0))}^2 + V + H_\mathrm{rad} .
\end{align*}
Under the dipole approximation, the estimate of the exponential functional becomes considerably simpler.
Let
\begin{align*}
    \tilde W_\mu
    = \int_0^T \mathrm{j}_s \tilde{\varphi} \, \mathrm{d}B_s^\mu ,
    \qquad \mu = 1, 2, 3 . 
\end{align*}
We consider the quadratic form
\begin{align*}
    X_\mathrm{dip} = {X_\mathrm{dip}}(T)
    = \sum_{\mu, \nu = 1}^3 {(\tilde W_\mu, d_{\mu\nu} \tilde W_\nu)}_{L^2}. 
\end{align*}
Define the compact operator $D_T$ on $L^2([0,T]; \mathbb{R}^3)$ by
\begin{align*}
    {{(D_T f)}_\mu}(s)
    = \sum_{\nu = 1}^3 \int_0^T D_{\mu \nu}(s,t) {f_\nu}(t) \, \mathrm{d}t ,
\end{align*}
where
\begin{align*}
    D_{\mu \nu}(s, t)
    = {(\mathrm{j}_s \tilde{\varphi}, d_{\mu \nu} \mathrm{j}_t \tilde{\varphi})}_{L^2} .
\end{align*}
Since $d_{\mu \nu}$ commutes with $\mathrm{j}_s$,
\begin{align*}
    D_{\mu \nu}(s, t)
    = \int_{\mathbb{R}^3}
        \frac{\abs{\hat{\varphi}(k)}^2}{\omega(k)}
        e^{- \abs{s - t} \omega(k)}
        d_{\mu \nu}(k)
    \, \mathrm{d}k .
\end{align*}
Then we obtain that
\begin{align*}
    \mathbb{E}\left[
        e^{\varepsilon \alpha^2 X_\mathrm{dip}}
    \right]
    ={\det(\mathbbm{1} - 2 \varepsilon \alpha^2 D_T)}^{- 1 / 2},
    \qquad \varepsilon \geq 0 ,
\end{align*}
where $\det_F$ denotes the Fredholm determinant.
More explicitly, if ${\{\lambda_n\}}_{n \geq 1}$ are the eigenvalues of the compact operator $D_T$, counted with multiplicity, then
\begin{align*}
    \det(\mathbbm{1} - 2 \varepsilon \alpha^2 D_T)
    = \prod_{n = 1}^{\infty} (1 - 2 \varepsilon \alpha^2 \lambda_n) . 
\end{align*}
Since $\omega$ and $\hat{\varphi} $ are radial, we see that
\begin{align*}
    D_{\mu \nu}(s, t)
    = \frac{2}{3} \delta_{\mu\nu} {C_T}(s, t) ,
\end{align*}
where
\begin{align*}
    {C_T}(s, t)
    = \int_{\mathbb{R}^3} 
        \frac{\abs{\hat{\varphi}(k)}^2}{\omega(k)}
        e^{- \abs{s - t} \omega(k)}
    \, \mathrm{d}k .
\end{align*}
If $C_T$ also denotes the corresponding integral operator on $L^2([0, T])$, then
\begin{align*}
    D_T
    = \frac{2}{3} {C_T} \otimes \mathbbm{1}_{3 \times 3} .
\end{align*}

\begin{lemma}\label{lem:16} %
    Suppose that
    \begin{align*}
        0 \leq \varepsilon
        < \frac{3}{8 \alpha^2 \norm{\hat{\varphi} / \omega}^2} .
    \end{align*}
    Then there exists a constant $c_\varepsilon > 0$, independent of $T$, such that
    \begin{align*}
        \mathbb{E}\left[
            e^{\varepsilon \alpha^2 X_\mathrm{dip}}
        \right]
        \leq e^{c_\varepsilon T}
    \end{align*}
    for all $T>0$.
    More precisely, one can take
    \begin{align}\label{eq:lem16} %
        c_\varepsilon
        = \frac{
            2 \varepsilon \alpha^2 \norm{\hat{\varphi} / \sqrt{\omega}}^2
        }{
            1 - 8 \varepsilon \alpha^2 / 3 \; \norm{\hat{\varphi} / \omega}^2
        } .
    \end{align}
\end{lemma}
\begin{proof}
    We have
    \begin{align*}
        \mathbb{E}\left[
            e^{\varepsilon \alpha^2 X_\mathrm{dip}}
        \right]
        = {\det(\mathbbm{1} - \frac{4 \varepsilon \alpha^2}{3} C_T)}^{- 3 / 2} .
    \end{align*}
    It is straightforward to see that
    \begin{align*}
        \norm{C_T}_{L^2}
        \leq \sup_{0 \leq s \leq T} \int_0^T {C_T}(s, t) \, \mathrm{d}t
        \leq 2 \int_{\mathbb{R}^3}
            \frac{\abs{\hat{\varphi}(k)}^2}{{\omega(k)}^2} \, \mathrm{d}k
        = 2\norm{\frac{\hat{\varphi}}{\omega}}^2 .
    \end{align*}
    Hence,
    \begin{align*}
        \frac{4 \varepsilon \alpha^2}{3} \norm{C_T}
        \leq \frac{8 \varepsilon \alpha^2}{3} \norm{\frac{\hat{\varphi}}{\omega}}^2
        < 1 .
    \end{align*}
    Set $a = \frac{4 \varepsilon \alpha^2}{3}$.
    Since $C_T \geq 0$ and $a \norm{C_T} < 1$, we have
    \begin{align*}
        - \log(\mathbbm{1} - a C_T)
        = \sum_{n = 1}^{\infty} \frac{a^n}{n} C_T^n
    \end{align*}
    and therefore
    \begin{align*}
        - \Tr \log(\mathbbm{1} - a C_T)
        = \sum_{n = 1}^{\infty} \frac{a^n}{n} \Tr C_T^n
        \leq \sum_{n = 1}^{\infty} a^n \norm{C_T}^{n - 1} \Tr C_T
        = \frac{a}{1 - a \norm{C_T}} \Tr C_T .
    \end{align*}
    Moreover,
    \begin{align*}
        \Tr C_T
        = \int_0^T {C_T}(s, s) \, \mathrm{d}s
        = T \norm{\hat{\varphi} / \sqrt{\omega}}^2 .
    \end{align*}
    Consequently,
    \begin{align*}
        \log
            \mathbb{E}\left[
                e^{\varepsilon \alpha^2 X_\mathrm{dip}}
            \right]
        = - \frac{3}{2} \Tr \log(\mathbbm{1} - a C_T)
        \leq \frac{3}{2} \frac{a}{1 - a \norm{C_T}} \Tr C_T
        \leq \frac{
            2 \varepsilon \alpha^2 \norm{\hat{\varphi} / \sqrt{\omega}}^2
        }{
            1 - 8 \varepsilon \alpha^2 / 3 \; \norm{\hat{\varphi} / \omega}^2
        }
        T .
    \end{align*}
    Thus the lemma follows. 
\end{proof}

\subsection{Lower bounds of exponential decay by Agmon distance}
We shall show  several propositions regarding the Agmon distance and related quantities.
\begin{assumption}\label{ass:17} %
{\rm Let $V$ satisfy the following properties.
    \begin{enumerate}
        \item[(1)] $V$ is continuous.
        \item[(2)] $V(x) \geq 0$ for all $x \in \mathbb{R}^3$, and $\sup_{\abs{x - y} \leq 1} V(y) > 0$ for all $x \in \mathbb{R}^3$.
        \item[(3)] $\lim_{\abs{x} \to + \infty} V(x) = + \infty$.
    \end{enumerate}
}\end{assumption}
We set
\begin{align*}
    V_{\sup}(x) = \sup_{\abs{x - y} \leq 1} V(y) .
\end{align*}
Under \cref{ass:17}, since $V$ is continuous on $\mathbb{R}^3$, it is locally bounded and belongs to the local Kato class.
The requirement $V_{\sup}(x) > 0$ ensures that there is no extended region in which the particle behaves as a free particle.
Finally, the condition $\lim_{\abs{x} \to + \infty} V(x) = + \infty$ implies that the potential confines the particle, preventing it from escaping to spatial infinity.
The above assumptions are natural in the context of Agmon theory.
In particular, they guarantees that the Agmon distance is complete.
To define the Agmon distance, we define the following path space:
\begin{align*}
    \mathcal{P}_{x \to y} &= \{\gamma \in C^1([0, T]; \mathbb{R}^3) \mid \gamma_0 = x, \gamma_T = y\}. 
\end{align*}
Let
\begin{align*}
    \mathscr{G}(\gamma)
    = \int_0^T \sqrt{2 V_{\sup}(\gamma_t)} \abs{\dot{\gamma}_t}
    \, \mathrm{d}t .
\end{align*}
The integrand of $\mathscr{G}(\gamma)$ can be regarded as the geometric mean of
$V_{\sup}(\gamma_t)$ and $1 / 2 \; \abs{\dot\gamma_t}^2$. 
The quantity $\mathscr{G}^*(x)$ is independent of the choice of $T$, since $\mathscr{G}(\gamma)$ is invariant under reparametrizations of the path.
Correspondingly, we introduce the arithmetic mean by
\begin{align*}
    \mathscr{A}(\gamma, T)
    = \int_0^T
        \left(V_{\sup}(\gamma_t) + \frac{1}{2} \abs{\dot{\gamma}_t}^2\right)
    \, \mathrm{d}t . 
\end{align*}
By the arithmetic and geometric inequality, 
$\mathscr{A}(\gamma, T) \geq \mathscr{G}(\gamma)$ holds.

For the sake of the subsequent analysis, we first approximate the potential by a smooth function. 
Let $V_{\sup} = w$. 
For any $b > 0$ there exists smooth $Y \in C^{\infty}(\mathbb{R}^3)$  such that 
\begin{align}\label{eq:Y} %
    (1 - b) w(x) \leq Y(x) \leq (1 + b) w(x).
\end{align}
Set 
\begin{align*}
    w_b = \frac{1}{1 - b} Y .
\end{align*}
In what follows, we work with $w_b$ in place of $w$.
Let
\begin{align*}
    {\mathscr{G}_b}(\gamma)
    &= {\mathscr{G}_b}(\gamma, T)
    = \int_0^T
        \sqrt{2 {w_b}(\gamma_s)}
        \abs{\dot{\gamma_s}}
    \, \mathrm{d}s , \\
    {\mathscr{A}_b}(\mathrm{q}, \tau)
    &= \int_0^{\tau} \left(
        {w_b}(\mathrm{q}_s)
        + \frac{1}{2} \abs{\dot{\mathrm{q}_s}}^2
    \right) \, \mathrm{d}s .
\end{align*}
We now relate the minimization problem for the Agmon length functional $\mathscr{G}_b$ to a variational problem associated with the energy functional $\mathscr{A}_b$. 
We shall compare
\begin{align*}
    &\inf\{
        \mathscr{G}_b(\gamma)
        \mid 
        \gamma_0 = 0 , 
        \gamma_T = x , 
        \gamma \in C^{\infty}([0, T]; \mathbb{R}^3)
    \}, \\
    &\inf\{
        {\mathscr{A}_b}(\mathrm{q}, \tau)
        \mid 
        \tau > 0, 
        \mathrm{q}_0 = x , 
        \mathrm{q}_\tau = 0 , 
        \mathrm{q} \in C^{\infty}([0, \tau]; \mathbb{R}^3)
    \} .
\end{align*}
Let $\gamma^* \in C^{\infty}([0, T]; \mathbb{R}^3)$ be the minimizer of ${\mathscr{G}_b}(\gamma)$ such that $\gamma^*_0 = 0$ and $\gamma^*_T = x$.  
The existence of a minimizer is proven in \cref{lem:a3}.  
There exists a minimizer $(\mathrm{q}^*, \tau^*) \in C^{\infty}([0, T]; \mathbb{R}^3) \times (0, + \infty)$ of ${\mathscr{A}_b}(\mathrm{q}, T)$. 
Both minimizers are coincide:
\begin{align}\label{eq:eq} %
    {\mathscr{G}_b}(\gamma^*)
    = {\mathscr{A}_b}(\mathrm{q}^*, \tau^*).
\end{align} 
This is shown in \cref{lem:a5}. 
Strictly speaking, the parametrization of the minimizer $\gamma^*$ connecting $0$ and $x$ depends on $T$.
However, since $\mathscr{G}_b$ is invariant under reparametrization, the set
\begin{align*}
    \{\gamma^*_t \mid 0 \leq t \leq T\}
\end{align*}
is independent of $T$, though it depends on the endpoint $x$.
Note also that $\tau^* = \tau^*(x)$ depends on $x$ and may tend to infinity as $\abs{x} \to + \infty$.
Nevertheless, it can be controlled from above as follows.
Let $\delta > 0$ be arbitrary. By \cref{lem:a6}, there exists $R_{\delta, b} > 0$ such that
\begin{align}\label{eq:tau} %
    \tau^*
    \leq \delta
    \int_0^{\tau^*} \left(
        {w_b}(\mathrm{q}^*_s)
        + \frac{1}{2} \abs{\dot{\mathrm{q}}^*_s}^2
    \right) \, \mathrm{d}s
\end{align}
for all $\abs{x} \geq R_{\delta, b}$.

\begin{lemma}\label{lem:18} %
    Let $Z \in \mathbb{R}$.
    Suppose that
    \begin{align*}
        {(\mathbbm{1}, {\varphi_\mathrm{g}}(x))}_\mathcal{F}
        \geq C e^{- \mathscr{A}(\mathrm{q}, T) + Z T}
    \end{align*}
    for any $\mathrm{q} \in \mathcal{P}_{0 \to x}$. 
    Then for any $\varepsilon > 0$ there exists $R_\varepsilon > 0$ such that
    \begin{align*}
        {(\mathbbm{1}, {\varphi_\mathrm{g}}(x))}_\mathcal{F}
        \geq C e^{- (1 + \varepsilon \abs{Z}) \mathscr{G}(\gamma)}
    \end{align*}
    holds for any $\abs{x} > R_\varepsilon$ and for any $\gamma \in \mathcal{P}_{0 \to x}$.
\end{lemma}
\begin{proof}
    By \cref{eq:Y} we have
    \begin{align*}
        {(\mathbbm{1}, {\varphi_\mathrm{g}}(x))}_\mathcal{F}
        \geq C e^{- {\mathscr{A}_b}(\mathrm{q}, T) + Z T} .
    \end{align*}
    Putting $T = \tau^*$ and $\mathrm{q} =\mathrm{q}^*$, we have
    \begin{align*}
        {(\mathbbm{1}, {\varphi_\mathrm{g}}(x))}_\mathcal{F}
        \geq C
        e^{
            - {\mathscr{A}_b}(\mathrm{q}^*, \tau^*) - \abs{Z} \tau^*
        } .
    \end{align*}
    For $\abs{x} \geq R_{\delta, b}$ with some $R_{\delta, b}$, by \cref{eq:tau}
    \begin{align*}
        {(\mathbbm{1}, {\varphi_\mathrm{g}}(x))}_\mathcal{F}
        &\geq C
        e^{
            - (1 + \delta \abs{Z}) {\mathscr{A}_b}(\mathrm{q}^*, \tau^*)
        } .
    \end{align*}
    By \cref{eq:eq}, ${\mathscr{A}_b}(\mathrm{q}^*, \tau^*) = {\mathscr{G}_b}(\gamma^*)$, we obtain that 
    \begin{align*}
        {(\mathbbm{1}, {\varphi_\mathrm{g}}(x))}_\mathcal{F}
        \geq C
        e^{
            - (1 + \delta \abs{Z}) {\mathscr{G}_b}(\gamma^*)
        } 
        \geq C
        e^{
            - (1 + \delta \abs{Z}) {\mathscr{G}_b}(\gamma)
        } ,
    \end{align*}
    where $\gamma^*$ is a minimizer of ${\mathscr{G}_b}(\gamma)$ and $\gamma \in \mathcal{P}_{0\to x}$ is arbitrary. 
    By \cref{eq:Y} again we have 
    \begin{align*}
        {(\mathbbm{1}, {\varphi_\mathrm{g}}(x))}_\mathcal{F}
        &\geq C \exp(
            - (1 + \delta \abs{Z}) \sqrt{\frac{1 + b}{1 - b}} \mathscr{G}(\gamma, T)
        ) .
    \end{align*}
    Choose $\delta$ and $b$ such that $1 + \varepsilon \abs{Z} \geq (1 + \delta \abs{Z}) \sqrt{(1 + b) / (1 - b)}$.
    Then the proof is complete. 
\end{proof}

\begin{lemma}\label{lem:19} %
    Suppose that $1 / p + 1 / q = 1$, $p > 1$, $q > 1$ and $\gamma \in \mathcal{P}_{0 \to x}$. 
    We assume that $p$ and $q$ satisfy that 
    \begin{align*}
        \frac{q}{p}
        < \frac{3}{8 \alpha^2 \norm{\hat{\varphi} / \omega}^2} .
    \end{align*}
    Then
    \begin{align*}
        {(\mathbbm{1}, {\varphi_\mathrm{g}}(x))}_\mathcal{F} 
        &\geq e^{- \frac{p}{q} c_{q / p} T}
        e^{- \frac{p}{2} \int_0^T \abs{\dot{\gamma}_t}^2 \, \mathrm{d}t}
        {\mathbb{E}\left[
            \mathbbm{1}_{S_T}
            e^{- \frac{1}{p} \int_0^T (V(B_t + \gamma_t) - E_\alpha) \, \mathrm{d}t}
            {m(B_T)}^{1 / p}
            e^{\int_0^T \dot{\gamma}_t \cdot \, \mathrm{d}B_t}
        \right]}^p ,
    \end{align*}
    where $c_{q / p}$ is given by \cref{eq:lem16} and then
    \begin{align*}
        \frac{p}{q} c_{q / p}
        = \frac{
            2 \alpha^2 \norm{\hat{\varphi} / \sqrt{\omega}}^2
        }{
            1 - 8 (q / p) \alpha^2 / 3 \; \norm{\hat{\varphi} / \omega}^2
        } .
    \end{align*}
\end{lemma}
\begin{proof}
    Shifting the starting point of Brownian motion to the origin, we have
    \begin{align*}
        {(\mathbbm{1}, {\varphi_\mathrm{g}}(x))}_\mathcal{F} 
        \geq \mathbb{E}\left[
            \mathbbm{1}_{S_T}
            e^{- \int_0^T (V(B_t + x) - E_\alpha) \, \mathrm{d}t}
            e^{- X}
            m(B_T + x)
        \right] . 
    \end{align*}
    We have the lower bound:
    \begin{align*}
        {(\mathbbm{1}, {\varphi_\mathrm{g}}(x))}_\mathcal{F} 
        \geq \frac{
            {\mathbb{E}\left[
                \mathbbm{1}_{S_T}
                e^{- \frac{1}{p} \int_0^T (V(B_t + x) - E_\alpha) \, \mathrm{d}t}
                {m(B_T + x)}^{1 / p}
            \right]}^p
        }{
            {\mathbb{E}[e^{{\frac{q}{p}} X}]}^{p / q}} .
    \end{align*}
    By \cref{lem:16},
    \begin{align*}
        {\mathbb{E}\left[e^{\frac{q}{p} X}\right]}^{p / q}
        \leq e^{\frac{p}{q} c_{q / p}}. 
    \end{align*}
    Then
    \begin{align*}
        {(\mathbbm{1}, {\varphi_\mathrm{g}}(x))}_\mathcal{F}
        \geq e^{- \frac{p}{q} c_{q / p} T}
        {\mathbb{E}\left[
            \mathbbm{1}_{S_T}
            e^{- \frac{1}{p} \int_0^T (V(B_t + x) - E_\alpha) \, \mathrm{d}t}
            {m(B_T + x)}^{1 / p}
        \right]}^p .
    \end{align*}
    For $\gamma \in \mathcal{P}_{x \to 0}$, we define a new measure $d\hat{\mathcal{W}}$ on $(\mathcal{X}, \mathcal{F})$ by
    \begin{align*}
        \mathrm{d}\hat{\mathcal{W}}
        = \exp(
            \int_0^T \dot{\gamma}_t \cdot \, \mathrm{d}B_t
            - \frac{1}{2} \int_0^T \abs{\dot{\gamma}_t}^2 \, \mathrm{d}t
        ) \mathrm{d}\mathcal{W} .
    \end{align*}
    Here $\gamma_t = \gamma_t$.
    Then $\mathrm{d}\hat{\mathcal{W}}$ is a probability measure on $(\mathcal{X}, \mathcal{F})$.
    Let $\gamma \in \mathcal{P}_{x \to 0}$ and
    \begin{align*}
        \hat{B}_t
        &= B_t
        + \int_0^t \dot{\gamma}_s \, \mathrm{d}s
        = B_t + \gamma_t - x . 
    \end{align*}
    Then ${(\hat{B}_t)}_{t \geq 0}$ is also 3D-Brownian motion under $\mathrm{d}\hat{\mathcal{W}}$.
    By Girsanov's theorem, we obtain that
    \begin{align*}
        {(\mathbbm{1}, {\varphi_\mathrm{g}}(x))}_\mathcal{F} 
        &\geq e^{- \frac{p}{q} c_{q / p} T}
        {\mathbb{E}\left[
            \mathbbm{1}_{S_T}
            e^{- \frac{1}{p} \int_0^T (V(\hat B_t + x) - E_\alpha) \, \mathrm{d}t}
            {m(\hat{B}_T + x)}^{1 / p}
            e^{
                \int_0^T \dot{\gamma}_t \cdot \, \mathrm{d}B_t
                - \frac{1}{2} \int_0^T \abs{\dot{\gamma}_t}^2 \, \mathrm{d}t
            }
        \right]}^p \\
        &= e^{- \frac{p}{q} c_{q / p} T}
        e^{- \frac{p}{2} \int_0^T \abs{\dot{\gamma}_t}^2 \, \mathrm{d}t}
        {\mathbb{E}\left[
            \mathbbm{1}_{S_T}
            e^{- \frac{1}{p} \int_0^T (V(B_t + \gamma_t) - E_\alpha) \, \mathrm{d}t} 
            {m(B_T)}^{1 / p}
            e^{\int_0^T \dot{\gamma}_t \cdot \, \mathrm{d}B_t}
        \right]}^p .
    \end{align*}
\end{proof}

\begin{theorem}\label{thm:20} %
    Suppose that $p$ and $q$ satisfy the assumptions of \cref{lem:19}, and that $V$ satisfies \cref{ass:17}. 
    Then there exists a constant $C>0$ such that, for every $\varepsilon > 0$, there exists $R_\varepsilon > 0$ for which
    \begin{align*}
        \norm{{\varphi_\mathrm{g}}(x)}_\mathcal{F}
        \geq C e^{- (p + \varepsilon) \mathscr{G}(\gamma)}
    \end{align*}
    for all $\abs{x} > R_\varepsilon$ and all $\gamma \in \mathcal{P}_{0 \to x}$.
\end{theorem}
\begin{proof}
    Let $\lambda$ denote the infimum of the spectrum of the operator $- (1 / 2) \Delta$ on the unit ball in $\mathbb{R}^3$ with the Dirichlet boundary condition.
        Then    
        \begin{align*}
            \mathbb{E}[\mathbbm{1}_{S_T}] \geq e^{- \lambda T}.
        \end{align*}
    See~\cite[\text{Lemma 2.3}]{CS81}.
    From \cref{lem:2,lem:a8}, $m(\cdot) = {(\mathbbm{1}, {\varphi_\mathrm{g}}(\cdot))}_\mathcal{F}$ is continuous and strictly positive.
    So, there exists $\ell > 0$ such that
    \begin{align*}
        m(B_T) \geq \ell
    \end{align*}
    holds on $S_T$.
    By $V_{\sup}(x) \geq V(x)$ for any $x \in \mathbb{R}^3$ we have
    \begin{align*}
        {(\mathbbm{1}, {\varphi_\mathrm{g}}(x))}_\mathcal{F} 
        \geq \ell 
        e^{- \frac{p}{q} c_{q / p} T}
        e^{- \int_0^T (V_{\sup}(\dot{\gamma}_t) + \frac{p}{2} \abs{\dot{\gamma}_t}^2 - E_\alpha) \, \mathrm{d}t} 
        {\mathbb{E}[
            \mathbbm{1}_{S_T}
            e^{\int_0^T \dot{\gamma}_t \cdot \, \mathrm{d}B_t}
        ]}^p
    \end{align*}
    Due to the Jensen's inequality we also have
    \begin{align*}
        \mathbb{E}\left[
            \mathbbm{1}_{S_T}
            e^{\int_0^T \dot{\gamma}_t \cdot \, \mathrm{d}B_t}
        \right]
        \geq \mathbb{E}[\mathbbm{1}_{S_T}]
        \exp(
            \frac{
                \mathbb{E}[
                    \mathbbm{1}_{S_T}
                    \int_0^T \dot{\gamma}_t \cdot \, \mathrm{d}B_t
                ]
            }{
                \mathbb{E}[\mathbbm{1}_{S_T}]
            }
        )
    \end{align*}
    The Wiener measure and the indicator function $\mathbbm{1}_{S_T}$ are invariant under the transformation $B_t \mapsto - B_t$, so we have
    \begin{align*}
        \mathbb{E}\left[
            \mathbbm{1}_{S_T}
            \int_0^T \dot{\gamma}_t \cdot \, \mathrm{d}B_t
        \right]
        = \mathbb{E}\left[
            \mathbbm{1}_{S_T}
            \int_0^T \dot{\gamma}_t \cdot \mathrm{d}(- B_t)
        \right]
        = - \mathbb{E}\left[
            \mathbbm{1}_{S_T}
            \int_0^T \dot{\gamma}_t \cdot \, \mathrm{d}B_t
        \right]
        = 0 . 
    \end{align*}
    We obtain
    \begin{align}
        {(\mathbbm{1}, {\varphi_\mathrm{g}}(x))}_\mathcal{F} 
        &\geq \ell
        e^{- \frac{p}{q} c_{q / p} T}
        {\mathbb{E}[\mathbbm{1}_{S_T}]}^p
        e^{- \int_0^T (V_{\sup}(\dot{\gamma}_t) + \frac{p}{2} \abs{\dot{\gamma}_t}^2 - E_\alpha) \, \mathrm{d}t} \notag \\
        &\geq \ell
        e^{- \frac{p}{q} c_{q / p} T}
        e^{- \lambda p T}
        e^{- \int_0^T (V_{\sup}(\dot{\gamma}_t) + \frac{p}{2} \abs{\dot{\gamma}_t}^2) \, \mathrm{d}t + T E_\alpha} \notag \\
        &\geq \ell   
        e^{- \int_0^T (V_{\sup}(\dot{\gamma}_t) + \frac{p}{2} \abs{\dot{\gamma}_t}^2) \, \mathrm{d}t}
        e^{-T \abs{\frac{p}{q} c_{q / p} + \lambda - E_\alpha}} .
    \end{align}
    Here we used that $p > 1$.
    Then
    \begin{align*}
        {(\mathbbm{1}, {\varphi_\mathrm{g}}(x))}_\mathcal{F}
        \geq \ell   
        e^{- p \mathscr{A}(\gamma) - \abs{Z} T} .
    \end{align*}
    Here $Z = (p / q) c_{q / p} + \lambda - E_\alpha$. 
    Hence by \cref{lem:18} we conclude that 
    there exists $R_\varepsilon > 0$ such that 
    \begin{align}
        {(\mathbbm{1}, {\varphi_\mathrm{g}}(x))}_\mathcal{F} 
        \geq \ell e^{- (p + \varepsilon \abs{Z}) \mathscr{G}(\gamma)}
    \end{align}
    holds for any $\abs{x} > R_\varepsilon$ and any $\gamma \in \mathcal{P}_{0 \to x}$.  
    Replacing  $\ell$ by $C$ and $\varepsilon \abs{Z}$ by $\varepsilon$, we complete the proof.
\end{proof}

Although \cref{thm:20} may not be optimal, it has the advantage that, for sufficiently small coupling constants, $p$ can be chosen arbitrarily close to $1$, whereas for large coupling constants the condition can still be accommodated by taking $p$ sufficiently large.

\begin{corollary}
    Suppose that $p$ and $q$ satisfy the assumptions of \cref{lem:19}.
    Let \[V(x) \leq \frac{A^2}{2}\abs{x}^{2 n}\] for some $n > 0$.
    Then there exists a constant $C > 0$ such that, for any $\varepsilon > 0$, there exists $R_\varepsilon > 0$ such that
    \begin{align*}
        \norm{{\varphi_\mathrm{g}}(x)}_\mathcal{F}
        \geq C \exp(
            - \frac{A}{n + 1} (p + \varepsilon) \abs{x}^{n + 1}
        )
    \end{align*}
    for all $\abs{x} > R_\varepsilon$.
\end{corollary}
\begin{proof} 
Let $\gamma_t = \frac{t}{T}x$ in \cref{thm:20}. Then
\begin{align*}
    \norm{\varphi_\mathrm{g}(x)}_\mathcal{F}
    \geq C
    \exp\left(
        -(p+\varepsilon)
        \int_0^T
        \sqrt{2V_{\sup}(\gamma_t)}
        \abs{\dot{\gamma}_t}
        \,\mathrm{d}t
    \right)
\end{align*}
for $\abs{x}>C_\varepsilon$, where $C_\varepsilon>0$ is some
constant. We note that
$
    V_{\sup}(x)\leq \frac{A^2}{2}(1+\abs{x})^{2n}$. 
Moreover,
\begin{align*}
    \int_0^T
    \sqrt{2V_{\sup}(\gamma_t)}
    \abs{\dot{\gamma}_t}
    \,\mathrm{d}t
    &\leq
    \int_0^T
    A\sqrt{\left(1+\abs{\frac{t}{T}x}\right)^{2n}}
    \abs{\frac{x}{T}}
    \,\mathrm{d}t
    \\
    &=
    \frac{A}{n+1}
    \left\{
        (1+\abs{x})^{n+1}-1
    \right\}
    \\
    &\leq
    \frac{A}{n+1}
    (1+\varepsilon)\abs{x}^{n+1}
\end{align*}
for $\abs{x}>C'_\varepsilon$, where $C'_\varepsilon>0$ is some
constant. Hence,
\begin{align*}
    \norm{\varphi_\mathrm{g}(x)}_\mathcal{F}
    \geq
    C\exp\left(
        -\frac{A}{n+1}
        (p+\varepsilon)(1+\varepsilon)
        \abs{x}^{n+1}
    \right)
\end{align*}
for all $
    \abs{x}>\max\{C_\varepsilon,C'_\varepsilon\}$. 
Since
$    (p+\varepsilon)(1+\varepsilon)
    =p+(p+1)\varepsilon+\varepsilon^2$, 
renaming $(p+1)\varepsilon+\varepsilon^2$ as $\varepsilon$
completes the proof of the corollary.
\end{proof} 

\section{Concluding remarks}\label{sec:5} %
An attempt to establish the spatial decay of the ground state of the full Pauli-Fierz model in terms of the Agmon distance leads naturally to a double stochastic integral, whose satisfactory control appears to be rather difficult.
This difficulty prevents a direct implementation of the Agmon distance approach. For this reason, in \cref{sec:3}  we adopt a different method to derive the desired decay estimate.
In this section, despite the difficulties described above, we derive a lower bound on the ground state of the full Pauli-Fierz Hamiltonian in terms of the Agmon distance.

Let $\eta > 1$. 
Suppose that
\begin{align}\label{eq:p} %
    p = 1 + \eta \triangle T, 
\end{align}
$1 / p + 1 / q = 1$ and $\gamma \in \mathcal{P}_{0 \to x}$.
Then $q / p < 1 / (\triangle T)$ is satisfied and
\begin{align*}
    {(\mathbbm{1}, {\varphi_\mathrm{g}}(x))}_\mathcal{F} 
    &\geq {C_3}\left(\frac{q}{p}\right)
    e^{- C_2 T}
    e^{- \frac{p}{2} \int_0^T \abs{\dot{\gamma}_t}^2 \, \mathrm{d}t} 
    {\mathbb{E}\left[
        \mathbbm{1}_{S_T}
        e^{- \frac{1}{p} \int_0^T (V(B_t + \gamma_t) - E_\alpha) \, \mathrm{d}t} 
        {m(B_T)}^{1 / p}
        e^{\int_0^T \dot{\gamma}_t \cdot \, \mathrm{d}B_t} 
    \right]}^p .
\end{align*}
Since $q / p = 1 / (\eta \triangle T)$, we can see that
\begin{align*}
    {C_3}\left(\frac{q}{p}\right)
    = {\left[
        \frac{
            2 \eta^2 e^{3 / 2}
        }{
            \sqrt{\eta^2 - 1} \left(\eta + \sqrt{\eta^2 - 1}\right)
        }
    \right]}^{\eta \triangle T / 6}
    \exp(
        \alpha_*^2 \norm{\frac{\hat{\varphi}}{\omega}}^2
    ) .
\end{align*}
Let 
\begin{align*}
    \kappa_\eta = \frac{\eta \triangle}{6}
    \log(
        \frac{
            2 \eta^2 e^{3 / 2}
        }{
            \sqrt{\eta^2 - 1} \left(\eta + \sqrt{\eta^2 - 1}\right)
        }
    ) .
\end{align*}
Hence
\begin{align*}
    {C_3}\left(\frac{q}{p}\right)
    = e^{\kappa_\eta T}
    e^{\alpha_*^2 \norm{\hat{\varphi} / \omega}^2} .
\end{align*}

\begin{proposition}
    Suppose \cref{eq:p} and that $V$ satisfies \cref{ass:17}.
    Then there exists a constant $C > 0$ such that, for any $\varepsilon > 0$, there exists $R_\varepsilon > 0$ such that
    \begin{align*}
        \norm{{\varphi_\mathrm{g}}(x)}_\mathcal{F}
        \geq C
        e^{- (1 + \varepsilon) \mathscr{G}(\gamma) - \varepsilon {\mathscr{G}(\gamma)}^2}
    \end{align*}
    for all $\abs{x} > R_\varepsilon$ and all $\gamma \in \mathcal{P}_{0 \to x}$.
\end{proposition}
\begin{proof}
    In the same way as the proof of \cref{thm:20}, we have 
    $\mathbb{E}[\mathbbm{1}_{S_T}] \geq e^{- \lambda T}$ and $m(B_T) \geq \ell$ on $S_T$.
    We have
    \begin{align*}
        {(\mathbbm{1}, {\varphi_\mathrm{g}}(x))}_\mathcal{F} 
        \geq \ell 
        \frac{e^{C_2 T}}{{C_3}(q / p)}
        e^{- \int_0^T (V_{\sup}(\dot{\gamma}_t) + \frac{p}{2} \abs{\dot{\gamma}_t}^2 - E_\alpha) \, \mathrm{d}t} 
        {\mathbb{E}[
            \mathbbm{1}_{S_T}
            e^{\int_0^T \dot{\gamma}_t \cdot \, \mathrm{d}B_t}
        ]}^p
    \end{align*}
    We obtain
    \begin{align*}
        {(\mathbbm{1}, {\varphi_\mathrm{g}}(x))}_\mathcal{F} 
        &\geq \ell
        \frac{e^{C_2 T}}{{C_3}(q / p)}
        {\mathbb{E}[\mathbbm{1}_{S_T}]}^p
        e^{- \int_0^T (V_{\sup}(\dot{\gamma}_t) + \frac{p}{2} \abs{\dot{\gamma}_t}^2 - E_\alpha) \, \mathrm{d}t} \\
        &\geq \ell e^{- \kappa_\eta T} e^{- \alpha_*^2 \norm{\hat{\varphi} / \omega}^2}
        e^{C_2 T}e^{- \lambda p T}
        e^{
            - \int_0^T (V_{\sup}(\dot{\gamma}_t) + \frac{p}{2} \abs{\dot{\gamma}_t}^2) \, \mathrm{d}t + T E_\alpha} \\
        &\geq
        \ell e^{- \alpha_*^2 \norm{\hat{\varphi} / \omega}^2}
        e^{- p \int_0^T (V_{\sup}(\dot{\gamma}_t) + \frac{1}{2} \abs{\dot{\gamma}_t}^2) \, \mathrm{d}t}
        e^{- \eta \triangle \lambda T^2}
        e^{-T \abs{C_2 + \lambda - E_\alpha + \kappa_\eta}} . 
    \end{align*}
    Here we used that $p = 1 + \eta \triangle T>1$. 
    Then
    \begin{align*}
        {(\mathbbm{1}, {\varphi_\mathrm{g}}(x))}_\mathcal{F}
        \geq \ell
        e^{- \alpha_*^2 \norm{\hat{\varphi} / \omega}^2} 
        e^{- (1 + \eta \triangle T) \mathscr{A}(\gamma) - \abs{Z} T - \eta \triangle \lambda T^2}.
    \end{align*}
    Here $Z = C_2 + \lambda - E_\alpha + \kappa_\eta$.
    Hence we conclude that there exists $R_\varepsilon > 0$ such that 
    \begin{align*}
        {(\mathbbm{1}, {\varphi_\mathrm{g}}(x))}_\mathcal{F} 
        \geq \ell
        e^{\alpha_*^2 \norm{\hat{\varphi} / \omega}^2}
        e^{
            - (1 + \varepsilon \abs{Z}) \mathscr{G}(\gamma)
            - \varepsilon \eta (1 + \lambda) \triangle {\mathscr{G}(\gamma)}^2
        }
    \end{align*}
    holds for any $\abs{x} > R_\varepsilon$ and any $\gamma \in \mathcal{P}_{0 \to x}$.  
    Replacing $\ell e^{- \alpha_*^2 \norm{\hat{\varphi} / \omega}^2}$ by $C$ and both $\varepsilon \abs{Z}$ and $\varepsilon \eta (1 + \lambda) \triangle$ by $\varepsilon$, we complete the proof.
\end{proof}

\appendix
\section{Abstract second quantization and phase rotations}\label{a1}
We briefly introduce second quantizations. 
Let $S$ be  a contraction operator from a Hilbert space $\mathcal{K} $ to a Hilbert space $\mathcal{K} '$. 
The second quantization $\Gamma(S)$ of $S$ is the contraction operator from the boson Fock space $\mathcal{F}(\mathcal{K})$ to the boson Fock space $\mathcal{F}(\mathcal{K'})$ defined by
\begin{align*}
    \Gamma(S) = \mathbbm{1} \oplus \bigoplus_{n = 1}^{\infty} [\otimes_{\mathrm{sym}}^{n} S] .
\end{align*}
Let $h$ be a closed operator acting in $\mathcal{K}$ and we define the closed operator $h^{(n)}$ acting on $\otimes^n_{\mathrm{sym}} \mathcal{K}$ by 
\begin{align*}
    h^{(n)}
    = \overline{
        \sum_{k = 1}^n
            \underbrace{\mathbbm{1} \otimes \cdots \otimes \stackrel{k}{h} \otimes \cdots \otimes \mathbbm{1}}_{n \mathrm{- fold}}
    } . 
\end{align*}
Hence the differential second quantization of $h$ is a closed operator acting in the boson Fock space $\mathcal{F}(\mathcal {K})$ defined by  
\begin{align*}
    \mathrm{d}\Gamma(h) = 0 \oplus \bigoplus_{n = 1}^{\infty} h^{(n)}.
\end{align*}
The differential second quantization of the identity $\mathbbm{1}$ is the number operator $N = \,\mathrm{d}\Gamma(\mathbbm{1})$. 
The second quantization $\Gamma(S)$ and the differential second quantization  $\,\mathrm{d}\Gamma(h)$ are related by
\begin{align*}
    e^{i t \mathrm{d}\Gamma(h)} = \Gamma(e^{i t h}) ,
    \qquad t \in \mathbb{R}
\end{align*}
for any self-adjoint operator $h$.
Thus
\begin{align*}
    \frac{\mathrm{d}}{\mathrm{d}t} \Gamma(e^{i t h}) \lceil_{t = 0} = i \, \mathrm{d}\Gamma(h) .
\end{align*}
We fix a Hilbert space $\mathcal{K}$ and consider a self-adjoint operator $h$ acting on $\mathcal{K}$. 
Then the commutation relations hold: 
\begin{align}\label{eq:ccr} %
    [\mathrm{d}\Gamma(h), a(f)] = - a(h f),
    \qquad [\mathrm{d}\Gamma(h), a^\dag(f)] = a^\dag(h f)
\end{align}
for $f \in D(h)$.
We define the phase rotation by $\Theta = e^{i \frac{\pi}{2} N}$.
This is the second quantization of the fractional Fourier transform on $L^2(\mathbb{R}^3)$.
It is unitary and satisfies
\begin{align*}
    \Theta = e^{i \frac{\pi}{2} N} = \Gamma\left(e^{i \frac{\pi}{2}} \mathbbm{1} \right) = \Gamma(i \mathbbm{1}) .
\end{align*}
Let
\begin{align*}
    \phi(f) = \frac{1}{\sqrt2}(a^\dag(f)+a(\bar f)) ,
    \qquad \pi(f) = \frac{i}{\sqrt2}(a^\dag(f)-a(\bar f)) .
\end{align*}
\begin{lemma}\label{lem:1} %
    Let $f \in \mathcal{K}$.
    Then
    \begin{align*}
        e^{- i \frac{\pi}{2} N} \phi(f) e^{i \frac{\pi}{2} N} = -\pi(f) .
    \end{align*}
\end{lemma}
\begin{proof}
By \cref{eq:ccr}, we have
$
    e^{- i \frac{\pi}{2} N} a(f) e^{i \frac{\pi}{2} N} =  i a(f)$ and 
    $e^{- i \frac{\pi}{2} N} {a^\dag}(f) e^{i \frac{\pi}{2} N} =- i {a^\dag}(f)$.
Then the lemma follows. 
\end{proof}

\section{Continuity of $m(x) = {(\mathbbm{1}, {\varphi_\mathrm{g}}(x))}_\mathcal{F}$}\label{sec:a2} %
This is a minor modification of~\cite[\text{Theorem 4.111}]{LHB26}.
\begin{lemma}\label{lem:a7} %
    Let $V = V_{+} - V_{-}$ with $V_{\pm} \in \mathcal{K}(\mathbb{R}^3)$.
    Then $m(x)$ is continuous in $x$.
\end{lemma}
\begin{proof}
    Let $m(y, T) = {(\mathbbm{1}, e^{- T (H_\alpha - E_\alpha)}\varphi_\mathrm{g}(y))}_\mathcal{F}$.
    Let $P_u = e^{- u (- \Delta / 2)}$.
    Note that $m(\cdot, T) \in L^{\infty}(\mathbb{R}^3)$.
    We have 
    \begin{align*}
        (P_u m(\cdot, T - u))(x)
        = \mathbb{E}^x[
            e^{-\int_u^T (V(B_s) - E_\alpha) \, \mathrm{d}s} F_u
        ] ,
        \quad u < T ,
    \end{align*}
    where $F_u = {(\mathbbm{1}, e^{- i \alpha A_\mathrm{E}(\oplus \int_u^T \tilde \rho(\cdot - B_s) \,\mathrm{d}B_s)} \mathrm{J}_{T - u}{\varphi_\mathrm{g}}(B_T))}_\mathcal{E}$.
    The map $x \mapsto (P_u m(\cdot, T - u))(x)$ is continuous by the smoothing effect of $P_u$.
    We have 
    \begin{align*}
        m(x, T) - m(y, T)
        &= m(x, T)-(P_u m(\cdot, T - u))(x) \\
        &+ (P_u m(\cdot, T - u))(x) - (P_u m(\cdot, T - u))(y) \\
        &+ (P_u m(\cdot, T - u))(y) - m(y, T) .
    \end{align*}
    It suffices to show that 
    $m(x, T) - (P_u m(\cdot, T - u))(x) \to 0$ uniformly in $x$ as $u \to 0$. 
    We have
    \begin{align*}
        &m(x, T) - (P_u m(\cdot, T - u))(x) \\
        &= \mathbb{E}^x[
            (
                e^{- \int_0^T (V(B_s) - E_\alpha) \, \mathrm{d}s}
                - e^{-\int_u^T (V(B_s) - E_\alpha) \, \mathrm{d}s}
            ) F_0
        ]
        + \mathbb{E}^x[
            e^{- \int_u^T (V(B_s) - E_\alpha) \, \mathrm{d}s}
            ( F_0 - F_u )
        ] .
    \end{align*}
    Using the Markov property of Brownian motion ${(B_t)}_{t \geq 0}$, we see that 
    \begin{align*}
        &\abs{
            \mathbb{E}^x[
                (
                    e^{- \int_0^t (V(B_s) - E_\alpha) \, \mathrm{d}s}
                    - e^{- \int_u^t (V(B_s) - E_\alpha) \, \mathrm{d}s}
                )
                F_0
            ]
        } \\
        &= \mathbb{E}^x[
            e^{- \int_u^t (V(B_s) - E_\alpha) \, \mathrm{d}s}
            (
                e^{- \int_0^u (V(B_s) - E_\alpha) \, \mathrm{d}s} - 1
            ) F_0
        ] \\
        &\leq \mathbb{E}^x[
            e^{- \int_u^t (V(B_s) - E_\alpha) \, \mathrm{d}s}
        ]
        \sup_{z \in \mathbb{R}^3}
            \mathbb{E}^z\left[
                \abs{e^{- \int_0^u (V(B_s) - E_\alpha) \, \mathrm{d}s} - 1}
            \right]
            \norm{F_0}_{\infty} .
    \end{align*}
    Let $K = \sup_{\substack{x \in \mathbb{R}^3 \\ 0 \leq u \leq t}} \mathbb{E}^x[e^{- \int_u^t (V(B_s) - E_\alpha) \, \mathrm{d}s}]$.
    By Khasminskii's lemma and using that $V_{\pm} \in \mathcal{K}(\mathbb{R}^3)$, we also see that 
    \begin{align*}
        &\leq K
        \sup_{z \in \mathbb{R}^3}
            \mathbb{E}^z\left[
                \abs{e^{-\int_0^u (V(B_s) - E_\alpha) \, \mathrm{d}s} - 1}
            \right]
        \norm{F_0}_{\infty} \\
        &\leq K
        \sup_{z \in \mathbb{R}^3}
            \mathbb{E}^z\left[
                \abs{e^{\int_0^u (V_{-}(B_s) - E_\alpha) \, \mathrm{d}s} - 1}
                + \abs{1 - e^{\int_0^u (V_{+}(B_s) - E_\alpha) \, \mathrm{d}s}}
            \right]
            \norm{F_0}_{\infty} \\
        &\leq K
        \left(
            \frac{k_-}{1 - k_-}
            + \frac{k_+}{1 - k_+}
        \right)
        \norm{F_0}_{\infty} ,
    \end{align*}
    where $k_{\pm} = \sup_{z \in \mathbb{R}^3} \mathbb{E}^z\left[\int_0^u (V_\pm(B_s) - E_\alpha) \, \mathrm{d}s\right]$.
    Since $\lim_{u \to 0} k_{\pm} = 0$ by the fact $V_{\pm} \in \mathcal{K}(\mathbb{R}^3)$, we have 
    \begin{align*}
        \lim_{u \to 0}
        \abs{
            \mathbb{E}^x[
                (
                    e^{- \int_0^t (V(B_s) - E_\alpha) \, \mathrm{d}s}- e^{- \int_u^t (V(B_s) - E_\alpha) \, \mathrm{d}s}
                ) F_0
            ]
        }
        = 0 .
    \end{align*}
    On the other hand we have 
    \begin{align*}
        \abs{
            \mathbb{E}^x[
                e^{- \int_u^T (V(B_s) - E_\alpha) \, \mathrm{d}s}(F_0 - F_u)
            ]
        }
        \leq \sup_{\substack{x \in \mathbb{R}^3 \\ 0 \leq u \leq T}}
        {\left(
            \mathbb{E}^x[
                e^{- 2 \int_u^T (V(B_s) - E_\alpha) \, \mathrm{d}s}
            ]
        \right)}^{1 / 2}
        {\left(
            \mathbb{E}^x[\abs{F_0 - F_u}^2]
        \right)}^{1 / 2}.
    \end{align*}
    Since $\lim_{u \to 0} \sup_{x \in \mathbb{R}^3} \mathbb{E}^x[\abs{F_0 - F_u}^2] = 0$, we can conclude that 
    \begin{align*}
        \lim_{u \to 0} \sup_{x \in \mathbb{R}^3}
            \abs{
                \mathbb{E}^x[
                    e^{- \int_u^T (V(B_s) - E_\alpha) \, \mathrm{d}s}
                    F_0
                ]
                - \mathbb{E}^x[
                    e^{- \int_u^T (V(B_s) - E_\alpha) \, \mathrm{d}s}
                    F_u
                ]
            }
        = 0 . 
    \end{align*}
    Then $m(\cdot, T)$ is continuous for each $T > 0$. 
\end{proof}

\begin{lemma}\label{lem:a8} %
    Let $V$ be Kato-decomposable. 
    Then $m(x)$ is continuous in $x$, and $m$ is a continuous version of $(\mathbbm{1}, {\varphi_\mathrm{g}}(\cdot))$.
\end{lemma}
\begin{proof}
    Let $V \in \mathcal{K}$.
    Let ${V_R}(x) = {V_{+}}(x) {\mathbbm{1}_{\abs{x} \leq R}}(x) - {V_{-}}(x)$. 
    Hence
    \begin{align*}
        {m_R}(x)
        = \mathbb{E}^x[
            e^{- \int_0^T {V_R}(B_s) \, \mathrm{d}s}
            (\mathbbm{1}, e^{- i \alpha A(K_T)} {\varphi_\mathrm{g}}(B_T))
        ]
        e^{E_\alpha T}
    \end{align*}
    is continuous in $x$ by \cref{lem:a7}.
    We have 
    \begin{align*}
        \abs{m(x) - {m_R}(x)}
        &\leq 2 C e^{T E_\alpha}
        \mathbb{E}^x[
            e^{\int_0^T {V_{-}}(B_s) \, \mathrm{d}s}
            \mathbbm{1}_{\{\sup_{0 \leq s \leq T} \abs{B_s} > R\}}
        ] \\
        &\leq 2 C e^{T E_\alpha}
        {\left(
            \mathbb{E}^x[
                e^{2 \int_0^T {V_{-}}(B_s) \, \mathrm{d}s}
            ]
        \right)}^{1 / 2}
        {\left(
            \mathcal{W}^x \left(
                \left\{\sup_{0\leq s\leq T} \abs{B_s} > R\right\} 
            \right)
        \right)}^{1 / 2}.
    \end{align*}
    Let $M \subset \{x \in \mathbb{R}^3 \mid \abs{x} \leq R / 2\}$ be a compact set. 
    By L\'{e}vy maximal inequality we have 
    \begin{align*}
        \sup_{x \in M} \abs{m(x) - {m_R}(x)}
        &\leq 2 C e^{T E_\alpha}
        \sup_{x \in M}
            {\left(
                \mathbb{E}^x[
                    e^{2 \int_0^T {V_{-}}(B_s) \, \mathrm{d}s}
                ]
            \right)}^{1 / 2}
        \sup_{x \in M} {\mathcal{W}^x(\{ \abs{B_T} > R \})}^{1 / 2} .
    \end{align*}
    Since 
    $\mathcal{W}^x(\{ \abs{B_T} > R\}) \to 0$ as $R \to + \infty$, we conclude that $m_R$ locally uniformly converges to $m$ as $R\to0$, 
    which implies that $m$ is continous. 
    Finally since ${(\mathbbm{1}, {\varphi_\mathrm{g}}(x))}_\mathcal{F} = m(x)$ for a.e. $x \in \mathbb{R}^3$ and $m(\cdot)$ is continuous, the lemma follows.
\end{proof}

\section{Agmon distance}
We introduce a geometric framework for measuring the spatial exponential decay of the ground state. Our presentation is a slight modification of that in~\cite{agm82,CS81}.

\subsection{Geodesic distance by $V$}
Suppose \cref{ass:17}. 
Let us set
\begin{align*}
    w = V_{\sup} .
\end{align*}
$w$ is also continuous and satisfies that $w(x) \geq \varepsilon$ for all $x \in \mathbb{R}^3$ and $\lim_{\abs{x} \to + \infty} w(x) = + \infty$.
We fix $T > 0$. 
We define two $C^1$-path spaces:
\begin{align*}
    &\mathcal{P}^*
    = \{
        \mathrm{q} \in {C^1}([0, T]; \mathbb{R}^3)
        \mid \mathrm{q}_0 = x, \mathrm{q}_T = 0
    \} , \\
    &\mathcal{P}
    = \{
        \gamma \in {C^1}([0, T]; \mathbb{R}^3)
        \mid \gamma_0 = 0, \gamma_T = x
    \} .
\end{align*}
Let
\begin{align*}
    \mathscr{A}(\mathrm{q}, T)
    &= \int_0^T \left(w(\mathrm{q}_s) + \frac{1}{2} \abs{\dot{\mathrm{q}_s}}^2 \right) \, \mathrm{d}s ,
    \quad \mathrm{q} \in \mathcal{P}^* , \\
    \mathscr{G}(\gamma, T)
    &= \int_0^T \sqrt{2 w(\gamma_s)} \abs{\dot{\gamma}_s} \, \mathrm{d}s ,
    \quad \gamma \in \mathcal{P} .
\end{align*}
In the differential geometry $\mathscr{A}(\mathrm{q}, T)$ describes the energy and $\mathscr{G}(\gamma, T)$ the distance.
We set $\gamma^\mathrm{q}_s = \mathrm{q}_{T - s}$ for $\mathrm{q} \in \mathcal{P}^*$ and $\mathrm{q}^{\gamma}_s = \gamma_{T - s}$ for $\gamma \in \mathcal{P}$.
Then $\gamma^\mathrm{q} \in \mathcal{P}$, $\mathrm{q} ^\gamma \in \mathcal{P}^*$ and
\begin{align*}
    &\mathscr{G}(\gamma^\mathrm{q}, T)
    \leq \mathscr{A}(\mathrm{q}, T) , \\
    &\mathscr{G}(\gamma, T)
    \leq \mathscr{A}(\mathrm{q}^\gamma, T)
\end{align*}
follow for any $\mathrm{q} \in \mathcal{P}^*$ and $\gamma \in \mathcal{P}$ by the arithmetic and geometric inequality.
We are interested in the existence of a minimizer $\gamma^*$ of 
the length $\mathscr{G}(\gamma, T)$.
This kind of problem is very standard in the differential geometry and it is related to geodesic arcs.
$w$ is however not a $C^{\infty}$-function.
Then we shall approximate $w$ by a $C^{\infty}$-function.
\begin{lemma}\label{lem:a1} %
    For $b > 0$ there exists $Y \in C^{\infty}(\mathbb{R}^3)$ such that
    \begin{align*}
        (1 - b) w(x) \leq Y(x) \leq (1 + b) w(x) .
    \end{align*}
\end{lemma}
\begin{proof}
    Let ${\{f_n\}}_n$ be a partition of unity such that $f_n \in C_0^{\infty}(\mathbb{R}^3)$ and $\mathrm{supp} \, f_n \cap K$ is empty for all but finitely many $n$ for any compact set $K$. 
    Let $\rho \in C_0^{\infty}(\mathbb{R}^3)$ be such that $\rho(x) \geq 0$ and $\int_{\mathbb{R}^3} \rho(x) \, \mathrm{d}x = 1$.
    Let ${\rho_\varepsilon}(x) = \varepsilon^{-3} \rho(x / \varepsilon)$.
    Let $K_n$ be a compact set such that 
    $\mathrm{supp} \, \rho_n \subset K_n$.
    Set $r_n = 2^{- n} b \inf_{x \in K_n} \abs{w(x)}$. 
    Suppose that $0 < \varepsilon_n < 1 / n$ so that $\mathrm{supp} \, (\rho_{\varepsilon_n} * f_n w) \subset K_n$ and $\sup_{x} \abs{(\rho_{\varepsilon_n}* f_n w)(x) - f_n(x) w(x)} \leq r_n$. 
    Then
    \begin{align*}
        - r_n
        \leq (\rho_{\varepsilon_n} * f_n w)(x) - f_n(x) w(x)
        \leq r_n        
    \end{align*}
    and we have
    \begin{align*}
        f_n(x) w(x) - r_n
        \leq (\rho_{\varepsilon_n} * f_n w)(x)
        \leq f_n(x) w(x) + r_n
    \end{align*}
    Define $Y = \sum_{n = 1}^{\infty} (\rho_{\varepsilon_n} * f_n w)(x)$.
    We have
    \begin{align*}
        (1 - b) w(x) \leq Y(x) \leq (1 + b) w(x) .
    \end{align*}
\end{proof}

For $x, y \in \mathbb{R}^3$ we define the geodesic distance for $w$ by
\begin{align*}
    \varrho(x, y) 
    = \inf\{
        \mathscr{G}(\gamma, T)
        \mid
        \gamma \in {C^1}([0, T]; \mathbb{R}^3) ,
        \gamma_0 = x ,
        \gamma_T = y
    \} .
\end{align*}
$\varrho$ defines a metric on $\mathbb{R}^3$. 
Suppose \cref{ass:17}. 
Let $Y$ be given by \cref{lem:a1}. 
Set
\begin{align*}
    w_b = \frac{1}{1 - b} Y
\end{align*}
for simplicity.
Then $W(x) \leq {w_b}(x)$. 
Let
\begin{align*}
    {\mathscr{G}_b}(\gamma, T)
    = \int_0^T \sqrt{2 {w_b}(\gamma_s)} \abs{\dot{\gamma}_s} \, \mathrm{d}s .
\end{align*}
Then $w_b \geq \varepsilon / (1 - b) > 0$, $w_b \in C^{\infty}(\mathbb{R}^3)$ and $\lim_{\abs{x} \to + \infty} {w_b}(x) = +\infty$. 
Let
\begin{align*}
    {\varrho_b}(x, y)
    = \inf\{
        {\mathscr{G}_b}(\gamma, T)
        \mid
        \gamma \in {C^1}([0, T]; \mathbb{R}^3) ,
        \gamma_0 = x ,
        \gamma_T = y
    \}.
\end{align*}
Let ${\varrho_b}(0, X) = \inf_{x \in X} {\varrho_b}(0, x)$ be the distance from $0$ to $X$.
One point compactification of $\mathbb{R}^3$ is denoted by $\mathbb{R}^3 \cup \{\infty\}$.
We define the distance from $0$ to $\{\infty\}$ by
\begin{align*}
    {\varrho_b}(0, \{\infty\})
    = \sup_{
        \substack{K \subset \mathbb{R}^3 \\
            K : \mathrm{compact}
        }
    }
    {\varrho_b}(0, \mathbb{R}^3 \setminus K) .
\end{align*}

\begin{lemma}\label{lem:a2} %
    $\varrho_b$ is geodesically complete. 
\end{lemma}
\begin{proof}
    By~\cite[\text{Lemma A1.2}]{agm82}, $\varrho_b$ is geodesically complete if and only if ${\varrho_b}(x, \{\infty\}) = + \infty$ for all $x \in \mathbb{R}^3$.
    By the Delta inequality, ${\varrho_b}(x, \{\infty\}) = + \infty$ for all $x \in \mathbb{R}^3$ if and only if ${\varrho_b}(y, \{\infty\}) = + \infty$ for some $y \in \mathbb{R}^3$.
    We shall show that ${\varrho_b}(0, \{\infty\}) = + \infty$. 
    Let $N > 0$. 
    Since $\lim_{\abs{x} \to \infty} {w_b}(x) = + \infty$, there exists $\delta_N$ such that ${w_b}(x) > N$ for any $\abs{x} > \delta_N$. 
    We can suppose that $\delta_N - \delta_{N / 2} > 1$. 
    Let $B_{\delta_N}$ be the closed ball centered at $0$ with radius $\delta_N$.
    We see that
    \begin{align*}
        {\varrho_b}(0, \{\infty\}) \geq 
        {\varrho_b}(0, \mathbb{R}^3 \setminus B_{\delta_N}) .
    \end{align*}
    There exists $y_\varepsilon \in \mathbb{R}^3 \setminus B_{\delta_N}$ such that
    \begin{align*}
        {\varrho_b}(0, \{\infty\}) + \varepsilon
        \geq \inf\{
            \int_0^T \sqrt{2 {w_b}(\gamma^\varepsilon_s)} \abs{\dot{\gamma}^\varepsilon_s} \, \mathrm{d}s 
            \mid
            \gamma^\varepsilon \in C^1([0, T]; \mathbb{R}^3) , \gamma^\varepsilon_0 = x ,
            \gamma^\varepsilon_T = y_\varepsilon
        \}.
    \end{align*}
    Then there exists $\eta^\varepsilon \in {C^1}([0, T]; \mathbb{R}^3)$ such that $\eta^\varepsilon_0 = x$ and $\eta^\varepsilon_T = y_\varepsilon$ and 
    \begin{align*}
        {\varrho_b}(0, \{\infty\})
        \geq \int_0^T \sqrt{2 {w_b}(\eta^\varepsilon_s)} \abs{\dot{\eta}^\varepsilon_s} \, \mathrm{d}s - 2 \varepsilon .
    \end{align*}
    Let us define
    $T_N = \sup\{s \mid \abs{\eta^\varepsilon_s} \leq \delta_{N/2}\}$. 
    Hence
    \begin{align*}
        {\varrho_b}(0, \{\infty\})
        &\geq \int_{T_N}^T \sqrt{2 {w_b}(\eta^\varepsilon_s)} \abs{\dot{\eta}^\varepsilon_s} \, \mathrm{d}s - 2 \varepsilon
        \geq \sqrt{N} \int_{T_N}^T \abs{\dot{\eta}^\varepsilon_s} \,\mathrm{d}s - 2 \varepsilon \\
        &\geq \sqrt{N} (\delta_N - \delta_{N/2}) - 2 \varepsilon
        \geq \sqrt{N} - 2 \varepsilon .
    \end{align*}
    Since $N$ is arbitrary, the lemma is proven. 
\end{proof}

\begin{lemma}\label{lem:a3} %
    There exists a minimizer $\gamma^* \in C^{\infty}([0, T]; \mathbb{R}^3)$ of ${\mathscr{G}_b}(\gamma)$.
\end{lemma}
\begin{proof}
    By \cref{lem:a2}, the associated Riemannian manifold is geodesically complete.
    Hence, by the Hopf-Rinow theorem~\cite[\text{Theorem 6.13 and Corollary 6.15}]{lee97}, any two points can be joined by a length-minimizing geodesic.
    Since ${\mathscr{G}_b}(\gamma)$ is the corresponding length functional, it follows that a minimizer $\gamma^*$ exists.
\end{proof}

Let $g(x) = {({g_{ij}}(x))}_{1 \leq i, j \leq 3}$ be
\begin{align*}
    g(x) =
    \begin{pmatrix}
        {w_b}(x) & 0 & 0 \\
        0 & {w_b}(x) & 0 \\
        0 & 0 & {w_b}(x)
    \end{pmatrix}
    \qquad x \in \mathbb{R}^3 .
\end{align*}
Let $\gamma_s = (\gamma_{1, s}, \gamma_{2, s}, \gamma_{3, s})$ and $\dot{\gamma_{i,s}} = {(\mathrm{d} \gamma_i / \mathrm{d} s)}(s)$ for $1 \leq i \leq 3$.
Then
\begin{align*}
    \int_0^T \sqrt{{w_b}(\gamma_s)} \abs{\dot{\gamma}_s} \, \mathrm{d}s
    = \int_0^T \sqrt{\sum_{i, j = 1}^3 {g_{ij}}(\gamma_s) \dot{\gamma}_{i,s} \dot{\gamma}_{j, s}} \, \mathrm{d}s .
\end{align*}
The minimizer $\gamma^*$ in \cref{lem:a3} satisfies the geodesic equation: 
\begin{align}\label{eq:geo1} %
    \ddot{\gamma}_k + \sum_{i, j = 1}^3 {\Gamma_{ij}^k}(\gamma) \dot \gamma_i \dot \gamma_j = 0 ,
    \quad 1 \leq k \leq 3 .
\end{align}
Here $\Gamma_{i j}^k$ denotes the Christoffel symbol given by
\begin{align*}
    \Gamma_{i j}^k
    = \frac{1}{2}
    \sum_{l = 1}^3
        g^{k l}
        \left( 
            \frac{\partial g_{il}}{\partial x_j}
            + \frac{\partial g_{jl}}{\partial x_i}
            - \frac{\partial g_{ij}}{\partial x_l}
        \right) ,
\end{align*}
where ${g^{k l}}(x) = {({g^{-1}}(x))}_{kl} = (1 / {w_b}(x)) \delta_{k l}$.
Since ${w_b}(x) \geq \varepsilon / (1 - b)$, $g^{-1}$ exists for any $x \in \mathbb{R}^3$.
Note that $\Gamma_{i j}^k = \Gamma_{j i}^k$ by symmetries.

\begin{lemma}\label{lem:a4} %
    Let $\gamma^*$ be the minimizer of ${\mathscr{G}_b}(\gamma)$. 
    Then $\abs{\dot\gamma^*_s} \neq 0$ for $s \in [0, T]$.
\end{lemma}
\begin{proof}
    Let ${w_k}(x) = (\partial w_b / \partial x_k)(x)$ for $1 \leq k \leq 3$.
    It follows that
    \begin{align*}
        \Gamma^k_{nn}
        =
        \begin{cases}
            - \frac{1}{2} \frac{w_k}{w_b} & k \neq n, \\
            \frac{1}{2} \frac{w_k}{w_b} & k = n .
        \end{cases}
        .
    \end{align*}
    On the other hand when $k, n$ and $m$ are different from each other, $\Gamma^k_{nm} = 0$. 
    Moreover we see that 
    \begin{align*}
        \Gamma_{12}^1 = \Gamma_{21}^1 = \frac{1}{2} \frac{w_2}{w_b} , \quad &\Gamma_{13}^1 = \Gamma_{31}^1 = \frac{1}{2} \frac{w_3}{w_b} , \\ 
        \Gamma_{23}^2 = \Gamma_{32}^2 = \frac{1}{2} \frac{w_3}{w_b} , \quad &\Gamma_{21}^2 = \Gamma_{12}^2 = \frac{1}{2} \frac{w_1}{w_b} , \\ 
        \Gamma_{13}^3 = \Gamma_{31}^3 = \frac{1}{2} \frac{w_1}{w_b} , \quad &\Gamma_{23}^3 = \Gamma_{32}^3 = \frac{1}{2} \frac{w_2}{w_b} .
    \end{align*} 
    Together with them, \cref{eq:geo1} is written as 
    \begin{align}\label{eq:geo2} %
        \begin{dcases*}
            {w_b}(\gamma) \ddot{\gamma}_1
            +w _1(\gamma) \dot{\gamma}_1^{2}
            +w _2(\gamma) \dot{\gamma}_1 \dot{\gamma}_2
            +w _3(\gamma) \dot{\gamma}_1 \dot{\gamma}_3
            - \frac{1}{2} w _1(\gamma) \abs{\dot{\gamma}}^2
            =0,
            \\
            {w_b}(\gamma) \ddot{\gamma}_2
            +w _1(\gamma) \dot{\gamma}_1 \dot{\gamma}_2
            +w _2(\gamma) \dot{\gamma}_2^{2}
            +w _3(\gamma) \dot{\gamma}_2 \dot{\gamma}_3
            - \frac{1}{2} w _2(\gamma) \abs{\dot{\gamma}}^2
            =0,
            \\
            {w_b}(\gamma) \ddot{\gamma}_3
            +w _1(\gamma) \dot{\gamma}_1 \dot{\gamma}_3
            +w _2(\gamma) \dot{\gamma}_2 \dot{\gamma}_3
            +w _3(\gamma) \dot{\gamma}_3^{2}
            - \frac{1}{2} w _3(\gamma) \abs{\dot{\gamma}}^2
            =0 .
        \end{dcases*}
    \end{align}
    From \cref{eq:geo2} it follows that
    \begin{align*}
        \frac{\mathrm{d}}{\mathrm{d}s} {w_b}(\gamma)\abs{\dot{\gamma}}^2
        = \sum_{k = 1}^3
            \left( 
                {w_k}(\gamma) \abs{\dot{\gamma}}^2
                + 2 {w_b}(\gamma) \ddot{\gamma}_k
            \right)
            \dot{\gamma}_k
        = 0 ,
    \end{align*}
    which implies that 
    ${w_b}(\gamma_s) \abs{\dot\gamma_s}^2$ is a nonzero constant. 
    Since the endpoints are distinct, the minimizing geodesic $\gamma^*$ cannot be constant.
    Hence ${w_b}(\gamma^*_s) \abs{\dot\gamma^*_s}^2$ cannot vanish identically.
    Therefore the constant value of ${w_b}(\gamma^*_s) \abs{\dot\gamma^*_s}^2$ is strictly positive.
    In particular $\abs{\dot\gamma_s} \neq 0$ for $0 \leq s \leq T$ since ${w_b}(\gamma_s) > 0$.
\end{proof}

We now relate the minimization problem for the Agmon length functional $\mathscr{G}_b$ to a variational problem associated with the action functional $\mathscr{A}_b$.
For $\tau > 0$ and $\mathrm{q} \in C^{\infty}([0, \tau]; \mathbb{R}^3)$, define
\begin{align*}
    {\mathscr{A}_b}(\mathrm{q}, \tau)
    = \int_0^\tau \left({w_b}(\mathrm{q}_s) + \frac{1}{2} \abs{\dot{\mathrm{q}_s}}^2 \right) \, \mathrm{d}s .
\end{align*}
We shall compare
\begin{align*}
    &\inf\{
        {\mathscr{G}_b}(\gamma,T)
        \mid
        \gamma_0 = 0, \gamma_T = x, \gamma \in C^{\infty}([0, T]; \mathbb{R}^3)
    \} \\
    &\inf\{
        {\mathscr{A}_b}(\mathrm{q}, \tau)
        \mid
        \tau > 0, \mathrm{q}_0 = x, \mathrm{q}_\tau = 0, \mathrm{q} \in C^{\infty}([0, \tau]; \mathbb{R}^3)
    \} .
\end{align*}

\begin{lemma}\label{lem:a5} %
    Suppose \cref{ass:17}.
    Let $b > 0$ and let $\gamma^*$ be the minimizer of $\mathscr{G}_b$ given in \cref{lem:a3}.
    Then there exists a minimizer $(q^*, \tau^*) \in 
    C^{\infty}([0, \tau^*]; \mathbb{R}^3) \times \lbrack 0, \infty \rparen$ of $\mathscr{A}_b$ such that
    \begin{align*}
        {\mathscr{G}_b}(\gamma^*, T)
        = {\mathscr{A}_b}(q^*, \tau^*) .
    \end{align*}
    Moreover, if we define
    \begin{align*}
        {\gamma_*}_s = q^*_{\tau^* - s},
        \qquad 0 \leq s \leq \tau^* ,
    \end{align*}
    then
    \begin{align*}
        {\mathscr{G}_b}(\gamma^*, T)
        = \int_0^{\tau^*} \sqrt{2 {w_b}({\gamma_*}_s)} \abs{\dot{\gamma_*}_s} \, \mathrm{d}s
        = {\mathscr{A}_b}(q^*, \tau^*) .
    \end{align*}
\end{lemma}
\begin{proof}
    Note that ${\mathscr{G}_b}(\gamma, T)$ is invariant under reparametrizations of the path $\gamma$.
    Thus, for any bijective $C^1$ map $\varphi: [0, T] \to [0, T]$, setting $\gamma_{\varphi, s} = \gamma_{\varphi(s)}$, we have
    \begin{align*}
        {\mathscr{G}_b}(\gamma, T)
        = {\mathscr{G}_b}(\gamma_\varphi, T).
    \end{align*}
    Let $\gamma^*$ be the minimizer of ${\mathscr{G}_b}(\gamma, T)$.
    Define
    \begin{align*}
        \hat{\gamma}_s = \gamma^*_{T s} ,
        \qquad 0 \leq s \leq 1 .
    \end{align*}
    Then
    \begin{align*}
        \mathscr{G}_b(\gamma^*,T)
        = \int_0^1 \sqrt{2 {w_b}(\hat{\gamma}_s)} \abs{\dot{\hat\gamma}_s} \, \mathrm{d}s .
    \end{align*}
    Since $2 \sqrt{a b} \leq a + b$ with equality if and only if $a = b$, it follows that
    \begin{align*}
        \sqrt{2w_b (\gamma_s)}
        = \abs{\dot\gamma_s} ,
        \qquad s \in [0, 1] ,
    \end{align*}
    if and only if
    \begin{align*}
        \int_0^1 \sqrt{2 {w_b}(\gamma_s)} \abs{\dot\gamma_s} \, \mathrm{d}s
        = \int_0^1 \left({w_b}(\gamma_s) + \frac{1}{2} \abs{\dot{\gamma}_s}^2 \right) \, \mathrm{d}s .
    \end{align*}
    Since $w_b \in C^{\infty}(\mathbb{R}^3)$, ${w_b}(x) > 0$ for every $x \in \mathbb{R}^3$, $\hat{\gamma} \in C^{\infty}([0, 1]; \mathbb{R}^3)$, and $\abs*{\dot{\hat\gamma}_s} > 0$ by \cref{lem:a4}, the function
    \begin{align*}
        F(t)
        = \frac{\sqrt{2 {w_b}(\hat{\gamma}_t)}}{\abs*{\dot{\hat\gamma}_t}} ,
        \qquad 0 \leq t \leq 1 ,
    \end{align*}
    is positive and Lipschitz continuous.
    Hence the ordinary differential equation
    \begin{align*}
        \dot\varphi(s)
        = F(\varphi(s)) ,
        \qquad \varphi(0) = 0 ,
    \end{align*}
    admits a unique solution.
    Let $\tau^* > 0$ be the first time such that $\varphi(\tau^*) = 1$.
    Then
    \begin{align*}
        \varphi: [0, \tau^*] \to [0, 1]
    \end{align*}
    is a $C^1$ bijection satisfying $\dot{\varphi}(s) > 0$.
    Define
    \begin{align*}
        \hat{\gamma}_\varphi
        = \hat{\gamma} \circ \varphi.
    \end{align*}
    Then
    \begin{align*}
        \abs*{\dot{\hat{\gamma}}_{\varphi, s}}
        = \abs*{\dot{\hat{\gamma}}_{\varphi(s)}} \dot{\varphi}(s)
        = \sqrt{2 {w_b}(\hat{\gamma}_{\varphi, s})} .
    \end{align*}
    Therefore
    \begin{align*}
        {\mathscr{G}_b}(\gamma^*, T)
        &= \int_0^1 \sqrt{2 {w_b}(\hat{\gamma}_s)} \abs*{\dot{\hat\gamma}_s} \, \mathrm{d}s
        = \int_0^{\tau^*} \sqrt{2 {w_b}(\hat{\gamma}_{\varphi, s})}
        \abs*{\dot{\hat{\gamma}}_{\varphi, s}} \, \mathrm{d}s \\
        &= \int_0^{\tau^*} \left(
            {w_b}(\hat{\gamma}_{\varphi, s})
            + \frac{1}{2} \abs*{\dot{\hat{\gamma}}_{\varphi, s}}^2
        \right) \, \mathrm{d}s
        = {\mathscr{A}_b}(\mathrm{q}^{\hat{\gamma}_\varphi}, \tau^*).
    \end{align*}
    For any $\mathrm{q} \in \mathcal{P}$ and $s > 0$,
    \begin{align*}
        {w_b}(\mathrm{q}_s)
        + \frac{1}{2} \abs{\dot{\mathrm{q}_s}}^2
        \geq \sqrt{2 {w_b}(\mathrm{q}_s)} \abs{\dot{\mathrm{q}_s}} .
    \end{align*}
    Hence
    \begin{align*}
        {\mathscr{A}_b}(\mathrm{q}, \tau)
        \geq \int_0^\tau \sqrt{2 {w_b}(\mathrm{q}_s)} \abs{\dot{\mathrm{q}_s}} \, \mathrm{d}s
        \geq {\mathscr{G}_b}(\gamma^*, T)
        = {\mathscr{A}_b}(\mathrm{q}^{\hat{\gamma}_\varphi}, \tau^*).
    \end{align*}
    Thus $\mathrm{q}^{\hat{\gamma}_\varphi}$ is a minimizer of $\mathscr{A}_b$.
    Setting $\mathrm{q}^* = \mathrm{q}^{\hat{\gamma}_\varphi}$ and $\gamma_* = \hat{\gamma}_\varphi$, we obtain
    \begin{align*}
        {\mathscr{G}_b}(\gamma^*, T)
        = {\mathscr{A}_b}(\mathrm{q}^*, \tau^*),
    \end{align*}
    and the lemma follows.
\end{proof}

\subsection{Bound of $\tau^*$}
$\tau^* = \tau^*(x)$ depends of the end point $x$.
A bound of $\tau^*$ is also established.  
\begin{lemma}\label{lem:a6} %
It follows that 
    \begin{align*}
        \lim_{\abs{x} \to + \infty}
            \frac{
                \tau^*
            }{
                \int_0^{\tau^*} \left({w_b}(\mathrm{q}^*_s) + \frac{1}{2} \abs{\dot{\mathrm{q}}^*_s}^2 \right) \, \mathrm{d}s
            }
        = 0 .
    \end{align*}
    In particular for any $\delta > 0$ there exists $R > 0$ such that 
    \begin{align*}
        \tau^*
        \leq \delta
        \int_0^{\tau^*} \left(
            {w_b}(\mathrm{q}^*_s)
            + \frac{1}{2} \abs{\dot{\mathrm{q}}^*_s}^2
        \right) \, \mathrm{d}s
    \end{align*}
    for any $\abs{x} > R$.
\end{lemma}
\begin{proof}
    The proof of \cref{lem:a6} is due to~\cite[\text{Proposition 2.5 (i\hspace{-1.2pt}i\hspace{-1.2pt}i)}]{CS81}.
    Note that $(1 - b) w \leq Y \leq (1 + b) w$ with $w = V_{\sup}$,
    $Y \in C^{\infty}(\mathbb{R}^3)$, and $w_b = Y / (1 - b)$.
    Let
    \begin{align*}
        \rho(x)
        = \inf\{
            {\mathscr{A}_b}(\mathrm{q}, \tau)
            \mid
            \tau > 0 ,
            \mathrm{q}_0 = x ,
            \mathrm{q}_S = 0 ,
            \mathrm{q} \in C^{\infty}([0, \tau]; \mathbb{R}^3)
        \} .
    \end{align*}
    Then
    \begin{align*}
        \rho(x)
        = \int_0^{\tau^*} \left(
            {w_b}(\mathrm{q}^*_s)
            + \frac{1}{2} \abs{\dot{\mathrm{q}}^*_s}^2
        \right) \, \mathrm{d}s .
    \end{align*}
    Since
$
        \sqrt{2 {w_b}(\mathrm{q}^*_s)}
        = \abs{\dot{\mathrm{q}}^*_s}$,
    we have
    \begin{align*}
        \rho(x)
        = 2 \int_0^{\tau^*} {w_b}(\mathrm{q}^*_s) \, \mathrm{d}s .
    \end{align*}
    Since $w_b \geq V_{\sup}$, it follows that
$
        \rho(x) \geq 2 \delta \tau^*$, 
    where $\delta = \inf_x V_{\sup}(x) > 0$.
    Let
    \begin{align*}
        a(R) = \min_{\abs{x} \geq R} V_{\sup}(x) ,
    \end{align*}
    and let $s_R$ be the first time such that $\abs{\mathrm{q}^*_{s_R}} = R$.
    Set $z = \mathrm{q}^*_{s_R}$ and $\eta_s = \mathrm{q}^*_{s + s_R}$ for $0 \leq s \leq \tau^* - s_R$.
    Then $\eta_0 = z$, $\eta_{\tau^* - s_R} = 0$, and
$
        \sqrt{2 {w_b}(\eta_s)}
        = \abs{\dot{\eta_s}} $.
    Thus
    \begin{align*}
        \rho(z)
        = \int_0^{\tau^* - s_R} \left(
            {w_b}(\eta_s)
            + \frac{1}{2} \abs{\dot{\eta_s}}^2
        \right) \, \mathrm{d}s .
    \end{align*}
    In the same manner as above, we obtain
    \begin{align*}
        \sup_{\abs{y} = R}\rho(y)
        \geq \rho(z)
        \geq 2 \delta (\tau^* - s_R) .
    \end{align*}
    Moreover, since $\abs{\mathrm{q}^*_s} \geq R$ for $0 \leq s \leq s_R$, we have
$
        \rho(x)
        \geq 2 s_R a(R)$. 
    Therefore
    \begin{align*}
        \tau^*
        \leq s_R + \frac{1}{2 \delta} \sup_{\abs{y} = R} \rho(y)
        \leq \frac{1}{2 a(R)} \rho(x) + \frac{1}{2 \delta} \sup_{\abs{y} = R} \rho(y) .
    \end{align*}
    Hence
    \begin{align*}
        \frac{\tau^*}{\rho(x)}
        \leq \frac{1}{2 a(R)}
        + \frac{1}{2 \delta} \frac{\sup_{\abs{y} = R} \rho(y)}{\rho(x)} .
    \end{align*}
    Since $a(R )\to \infty$ as $R \to \infty$, we may choose $R$ sufficiently large so that $1 / (2 a(R)) < \varepsilon$.
    For this fixed $R$, letting $\abs{x} \to \infty$ gives
    \begin{align*}
        \limsup_{\abs{x} \to + \infty} \frac{\tau^*}{\rho(x)}
        \leq \varepsilon .
    \end{align*}
    Since $\varepsilon > 0$ is arbitrary, \cref{lem:a6} follows.
\end{proof}

\section*{Acknowledgments}
    We are grateful to Masao Hirokawa for his valuable advice and encouragement in both mathematics and physics.
    F.H. was partially supported by JSPS KAKENHI Grant Numbers JP25H00595 and JP26K21807.
    F.H. also acknowledges the warm hospitality of the University of Franche-Comt\'{e}, Besan\c{c}on, where part of this work was carried out.
    Y.T. was also supported by the Research Institute for Mathematical Sciences, an International Joint Usage/Research Center at Kyoto University, and by the WISE Program of Kyushu University funded by MEXT.
The authors also gratefully acknowledges the support of the Quantum and
Spacetime Research Institute (QuaSR), Kyushu University.%

\printbibliography%

\end{document}